\documentclass[12pt]{article}
\usepackage{amsfonts,amsmath,amsbsy,amssymb,amsthm,graphicx,latexsym,graphics,epsfig,color}
\usepackage{mathrsfs}
\usepackage{bm}
\usepackage{caption}
\usepackage{subcaption}
\usepackage{array}
\usepackage{booktabs}
\usepackage{tabularx}
\usepackage{algorithm}
\usepackage{algpseudocode}
\usepackage{longtable,color}
\usepackage{multirow,epsf}
\usepackage{setspace}
\usepackage{epstopdf}
\usepackage{float,wasysym}
\usepackage{color}
\usepackage{adjustbox}
\usepackage{enumerate} 
\usepackage{authblk} 

\usepackage{indentfirst}
\usepackage{graphicx}			
\usepackage{mathrsfs}
\usepackage{amsfonts}
\usepackage{amsmath,amssymb,mathrsfs,bm,latexsym,graphicx,color}		
\usepackage{amssymb}
\usepackage[dvipsnames]{xcolor}
\usepackage{colortbl}
\usepackage{natbib} \bibpunct{(}{)}{;}{a}{,}{,} 
\usepackage[normalem]{ulem}
\usepackage[colorlinks=true,urlcolor=blue,citecolor=blue,linkcolor=blue,bookmarks=true]{hyperref} 
\usepackage{titlesec}
\titleformat{\section}[block]{\large\bfseries}{\thesection}{1em}{}

\usepackage{multicol}
\usepackage{multirow}
\usepackage{tikz,circuitikz} 

\usepackage[T1]{fontenc}
\usepackage[utf8]{inputenc}

\allowdisplaybreaks

\newtheorem{proposition}{Proposition}

\usepackage{tikz,xcolor,hyperref}

\begin{document}
	\title{\bf \large  Reliability and Preventive Maintenance for Multivariate Gamma Degradation with Shared Random Effects and Cumulative Environmental Shock Damage}
\author{	  
	Jiale Niu\\ 
	{ College of Mathematics and Statistics }\\
	{Northwest Normal University, Lanzhou 730070, China}}


\maketitle

\begin{abstract}
	In industrial engineering, many highly reliable products have multiple performance characteristics that degrade simultaneously over time. Owing to common manufacturing factors and shared operating environments, their degradation processes are often statistically dependent. External environmental shocks may also accumulate into common damage and jointly contribute to system failure. To characterize such dependence, unit heterogeneity, and cumulative environmental damage, this study proposes a multivariate Gamma degradation model with a shared random effect and a shared environmental degradation process. Each degradation process consists of an individual Gamma process and a shared Gamma process, while the random effect captures latent heterogeneity across units and induces additional dependence. An explicit expression for system reliability is derived. For parameter estimation, an expectation maximization algorithm combined with variational inference is developed, and parametric bootstrap is used for uncertainty quantification. A preventive maintenance optimization approach is further established. Simulation studies and a coating degradation application demonstrate the effectiveness and practical applicability of the proposed methods.
\end{abstract}

\medskip 

\section{Introduction}\label{intro} 
In engineering applications, life testing has long been widely used to obtain product failure information and assess product reliability. However, with the rapid development of highly reliable and long life products, traditional life testing methods face increasingly significant limitations. On the one hand, for highly reliable products, a sufficient number of failure events are often difficult to observe within a limited testing period, resulting in a severe shortage of failure data. On the other hand, for high value products, conducting large scale life testing experiments usually requires substantial time and economic resources, making them impractical for many engineering applications. To overcome these challenges, degradation testing has gradually emerged as an effective alternative approach. By continuously monitoring performance characteristics (PCs) that are closely associated with product health conditions, degradation testing provides valuable information on the temporal evolution of product performance (\cite{hajiha2021degradation,hong2020nonparametric}). Compared with traditional reliability analysis methods based on failure data, degradation data can reveal the deterioration process of product performance at an earlier stage and contain richer information regarding product health states. Therefore, degradation based reliability analysis has been widely recognized as an important approach for improving the accuracy and efficiency of reliability assessment (\cite{chen2018uncertainty,ye2019strategic}).

Degradation models can generally be classified into two major categories: general path models and stochastic process based degradation models(\cite{ye2015stochastic}). General path models are essentially mixed effects regression models, in which fixed parameters and random parameters are used to describe the population degradation trend and unit heterogeneity, respectively. Specifically, \cite{fugate2017analyzing} proposed a Bayesian random effects spline regression model, where cubic splines with natural boundary conditions were employed to flexibly approximate individual degradation paths, and the probability distribution of spline coefficients was used to characterize unit heterogeneity. Based on the general path cumulative exposure model, \cite{ye2019strategic} further considered the correlation between initial product performance and degradation rates and developed a ranking based accelerated stress allocation strategy. Subsequently, \cite{xu2020bias} proposed a two stage approach for degradation data analysis, where a nonlinear mixed effects model was first used to fit individual degradation paths, and the lifetime distribution was then estimated by maximizing the asymptotic marginal distribution of pseudo lifetimes. In addition to general path models based on regression structures, stochastic process based degradation models have received considerable attention in both industry and academia because they can more naturally characterize the stochastic evolution of degradation processes and provide desirable mathematical properties with clear physical interpretations. Among them, Gamma processes, Wiener processes, and inverse Gaussian (IG) processes are three typical stochastic process models for degradation analysis. In particular, Gamma processes have been widely adopted in reliability modeling due to their nonnegative increments and monotonic increasing properties, which make them suitable for describing irreversible cumulative degradation processes, such as wear, corrosion, and crack propagation. Moreover, the mathematical tractability of Gamma processes enables their flexible extensions by incorporating covariates, random effects, and other stochastic factors. \cite{abdel1975gamma} was the first to describe cumulative wear as a Gamma process and investigated the lifetime distribution of a device when the degradation level first reached the failure threshold. Subsequently, to improve the applicability of Gamma processes under practical degradation conditions, \cite{bagdonavicius2001estimation} incorporated environmental factors and stress related covariates into Gamma degradation models and studied the estimation of reliability and degradation characteristics in a dynamic environment where shock intensity depended on the degradation state. To account for unit heterogeneity, \cite{lawless2004covariates} further introduced random effects into Gamma process models to characterize degradation differences among different units and used the model for lifetime prediction under failure thresholds. Considering the measurement errors commonly encountered in degradation observations, \cite{lu2013efficient} developed an efficient parameter estimation method for Gamma process models with measurement errors by combining the Genz transformation and quasi Monte Carlo methods. Furthermore, \cite{ling2019bayesian} constructed a two stage Gamma degradation model with random change points and employed Bayesian methods and stochastic expectation maximization methods to predict individual residual lifetimes. Other extensions of Gamma processes can be found in \cite{giorgio2018bayesian,song2020time} and \cite{li2022statistical}. In recent years, with the increasing reliability requirements of complex engineering products, degradation modeling under multiple influencing factors and unit heterogeneity has gradually become an important research topic. The aforementioned studies on Gamma processes have provided a solid foundation for incorporating covariates, random effects, and complex degradation mechanisms into reliability models. In comparison, Wiener processes and inverse Gaussian processes have also been widely applied in degradation based reliability analysis because they are capable of describing degradation paths with stochastic fluctuations. For further discussions on degradation models based on Wiener processes and inverse Gaussian processes and their applications, interested readers may refer to \cite{peng2015inverse,chen2015condition,zhai2018random,wang2021geometric}, and the references therein.

With the increasing functional complexity and reliability requirements of complex engineering products, a single performance indicator is often insufficient to fully characterize the overall health condition of a product. In practical applications, products are usually characterized by multiple performance characteristics (PCs), and their reliability depends on the joint degradation behavior of these PCs(\cite{yi2025statistical}). Therefore, degradation modeling with multiple performance characteristics has gradually become an important research topic in reliability analysis. \cite{pan2011reliability} developed a reliability model for products with multiple performance characteristics based on Gamma processes and Copula functions, where the statistical dependence among different degradation processes was characterized through their joint distribution. Subsequently, \cite{fang2020copula} further applied Copula methods to reliability analysis of degrading systems with dependent failure mechanisms to capture more flexible dependence structures among multiple degradation processes. In recent years, researchers have further investigated modeling approaches for dependence among multiple degradation processes and the effects of shared factors on degradation behaviors. \cite{zhai2023multivariate} proposed a multivariate stochastic degradation model that introduced a shared stochastic process to characterize the correlated degradation behaviors of multiple performance characteristics affected by common factors. Similarly, \cite{zhuang2025multivariate} incorporated common effects into inverse Gaussian processes and developed a multivariate degradation model with shared influence mechanisms for degradation based reliability prediction. These studies indicate that common factors may serve as an important source of dependence among multiple PC degradation processes.

However, in practical engineering applications, the dependence among PCs is often influenced by multiple common factors. On the one hand, degradation processes in multi component systems may be affected by shared material properties, manufacturing variations, and changes in operating conditions. On the other hand, environmental factors, such as temperature variations, ultraviolet radiation, chemical corrosion, and mechanical shocks, may simultaneously affect multiple PCs and introduce common cumulative damage effects. Therefore, describing the statistical dependence among degradation processes solely through Copula functions or covariance structures may be insufficient to fully characterize the correlated degradation behaviors induced by shared factors. In addition, variations in manufacturing processes, material properties, and operating environments among product units may lead to different degradation behaviors across units, resulting in unit heterogeneity. To address this issue, \cite{lawless2004covariates} introduced random effects into Gamma process models to characterize degradation differences among different units. Recently, \cite{song2025multivariate} further developed a multivariate Gamma degradation model based on hierarchical random effects to improve the characterization of heterogeneity in complex product degradation. However, existing studies usually consider PCs dependence, unit heterogeneity, or shared influencing factors separately. Research on a unified framework that simultaneously accounts for the joint effects of unit heterogeneity and shared environmental damage on the degradation processes of multiple PCs remains limited. Therefore, developing a degradation model that can characterize dependent degradation behaviors among multiple PCs while simultaneously considering unit heterogeneity and shared environmental damage remains an important issue in the reliability analysis of complex engineering products.

Motivated by the above observations, we propose a multivariate Gamma degradation model for products with multiple dependent PCs. In the proposed model, the degradation process of each PC is represented as the sum of an individual degradation component and a shared degradation component. Random effects are introduced to characterize unit heterogeneity. Meanwhile, since the random effects are shared among multiple PCs, they also capture part of the dependence among PCs induced by latent unit variability. In addition, the shared degradation component is introduced to describe the common environmental damage affecting multiple PCs. This modeling structure simultaneously accounts for unit to unit variability and the dependence among multiple degradation processes arising from different sources, thereby providing a more comprehensive characterization of multivariate degradation behaviors in complex engineering products. Based on the proposed model, statistical inference procedures are developed, and reliability assessment as well as preventive maintenance policy are investigated for complex products with dependent degradation processes.

The major contributions of this study are summarized as follows:

\begin{enumerate}
	\item [(1)]A multivariate Gamma degradation model is developed to characterize dependent degradation among multiple PCs within the same product unit through individual and shared degradation components, capturing the effects of latent unit characteristics and shared environmental damage.
	\item [(2)]Random effects are incorporated to account for unit heterogeneity across product units and shared latent unit characteristics among PCs within the same unit, thereby improving the characterization of complex degradation.
	\item [(3)]Statistical inference and reliability assessment methods are developed for the proposed model, where an EM algorithm combined with variational inference is employed for parameter estimation.
	\item [(4)]Preventive maintenance optimization is conducted based on the developed reliability model, and numerical studies and a real degradation data analysis are provided to demonstrate the effectiveness of the proposed method.
\end{enumerate}

The remainder of this manuscript is organized as follows. Section \ref{section2} presents the model formulation, where a multivariate Gamma degradation model is developed for products with multiple dependent PCs by incorporating individual degradation components, shared degradation components, and random effects. Section \ref{section3} derives the reliability function and analyzes the reliability characteristics of products with dependent degradation processes. Section \ref{section4} develops the statistical inference procedure for the proposed model, where an EM algorithm combined with variational inference is employed for parameter estimation. Section \ref{section5} investigates preventive maintenance strategies based on the proposed reliability model and determines the corresponding maintenance decisions. Section \ref{section6} conducts simulation studies to evaluate the performance of the proposed method under different experimental settings. Section \ref{section7} presents a real case study based on degradation data of coating materials to demonstrate the practical applicability of the proposed model. Section \ref{section8} concludes this manuscript and discusses potential directions for future research.

\section{Model formulation}\label{section2}
A stochastic process $\{X(t), t \ge 0\}$ is called the gamma process with scale parameter $\lambda > 0$ and shape function $\eta(t)$ if it satisfies the following conditions:

\begin{itemize}
	\item $X(t)$ has independent increments;
	\item For any $t > s \ge 0$, $X(t) - X(s)$ follows the gamma distribution $Ga(\eta(t) - \eta(s),\lambda)$, 
	\end{itemize}
	where $\eta(\cdot)$ is a monotone increasing function with $\eta(0) = 0$, and $\Gamma(x) = \int_{0}^{\infty} e^{-z} z^{x-1} \mathrm{d}z$ is the gamma function. The probability density function (PDF) of $Ga(\eta, \lambda)$ are defined as
	\begin{equation*}
		f_{Ga}(x; \eta,\lambda) = \frac{\lambda^\eta}{\Gamma(\eta)}e^{-\lambda x} x^{\eta-1}, \quad x > 0, \, \lambda > 0, \, \eta > 0.
	\end{equation*}
	With loss of generality, we assume $X(0)\equiv0$, then $X(t)$ follows $Ga(\eta(t), \lambda)$ with expectation $\eta(t)/\lambda$ and variance $\eta(t)/\lambda^2$.
	
	Many engineering systems are monitored through multiple performance characteristics (PCs), whose degradation collectively reflects the health condition of the system. A representative example is a hydraulic pump, whose health condition is commonly characterized by multiple performance indicators, including volumetric efficiency, flow rate, outlet pressure, vibration signals, and leakage-related measurements. These indicators reflect different aspects of pump degradation and are often affected by common wear mechanisms and operating conditions, resulting in statistically dependent degradation behaviors.
	
	Besides the intrinsic degradation caused by material wear and fatigue, hydraulic pumps are frequently subjected to external shock environments, such as pressure fluctuations, vibration loads, contaminant particles in the hydraulic fluid, and transient operating conditions. These shocks may not only generate immediate damage but also accelerate the subsequent degradation process. Moreover, their effects are often cumulative and persistent rather than instantaneous. Consequently, representing shock effects through a continuous stochastic degradation process is more realistic than modeling them as isolated damage increments.
	
	Motivated by the above considerations, we consider a unit characterized by $K$ performance characteristics, each of which undergoes a gradual degradation process over time. Let $X_k(t), \ k=1,2,\ldots,K$, denote the intrinsic degradation process associated with the $k$th PC. In addition to intrinsic degradation, the unit is exposed to an external shock environment, whose cumulative effect is represented by a common shock-induced degradation process $Z(t)$.
	
	The dependence among different PCs naturally arises from several practical considerations. First, common manufacturing variability—including material heterogeneity, production tolerances, and assembly imperfections—may jointly influence the degradation behaviors of multiple PCs. Second, all PCs operate under the same environmental and operational conditions, where temperature, humidity, vibration, and load may simultaneously affect their degradation rates. Third, exposure to a common shock environment introduces additional dependence, since a single shock may simultaneously impact multiple PCs. To capture these sources of dependence, a shared random-effect factor $\lambda$ is introduced. Conditional on $\lambda$, the intrinsic degradation processes $X_k(t)$ and the shock-induced degradation process $Z(t)$ are mutually independent. The latent factor $\lambda$ captures the overall degradation susceptibility of the unit and induces correlation among all degradation components.
	
	Under this framework, the total degradation process of the $k$th PC is represented as
	\begin{equation}\label{a21}
		\left\{
		\begin{aligned}
			Y_k(t) &= X_k(t) + Z(t),\\
			X_k(t)\mid \lambda &\sim \text{Ga}(\eta_k(t),\lambda),\\
			Z(t)\mid \lambda &\sim \text{Ga}(\eta_0(t),\lambda),\\
			\lambda &\sim \text{Ga}(\omega,\rho),
		\end{aligned}
		\right.
	\end{equation}
	where $\eta_k(t)$ and $\eta_0(t)$ are monotone increasing functions of time $t$, representing the cumulative effects of the intrinsic and shock-induced degradations, respectively. This additive construction allows the model to capture both the individual stochastic degradation of each PC and the common degradation induced by shared environmental factors and shocks, providing a realistic framework for reliability assessment and preventive maintenance planning.
	
	By the additive property of Gamma distributions, and noting that
	$X_k(t)$ and $Z(t)$ are conditionally independent given the shared
	random-effect factor $\lambda$, the total degradation process
	$Y_k(t)$ remains Gamma distributed. Specifically,
	\begin{align}\label{a10}
		Y_k(t)\mid\lambda
		\sim
		\mathrm{Ga}\bigl(\eta_k(t)+\eta_0(t),\lambda\bigr).
	\end{align}
	
	Therefore, the proposed model preserves the analytical tractability
	of Gamma-process-based degradation models, which facilitates subsequent
	reliability analysis and statistical inference. Based on this model
	structure, the following proposition can be established.
	
\begin{proposition}\label{pro1}
		For $k=1,\ldots,K$, the mean and variance of the degradation process $Y_k(t)$ are
		\[
		E\{Y_k(t)\}
		=
		\frac{\rho}{\omega-1}\left[ \eta_k(t)+\eta_0(t)\right] ,
		\]
		and
		\[
		\operatorname{Var}\{Y_k(t)\}
		=
		\frac{
			\rho^2
			(\eta_k(t)+\eta_0(t))
			\left[
			\eta_k(t)+\eta_0(t)+\omega-1
			\right]
		}{
			(\omega-1)^2(\omega-2)
		}.
		\]
		Moreover, for any $k_1\neq k_2$, the covariance between two total
		degradation processes is given by
		\[
		\operatorname{Cov}\{Y_{k_1}(t),Y_{k_2}(t)\}
		=
		\frac{\rho^2(\eta_{k_1}(t)+\eta_0(t))(\eta_{k_2}(t)+\eta_0(t))}{(\omega-1)^2(\omega-2)}+\frac{\rho^2\eta_0(t)}{(\omega-1)(\omega-2)}.
		\]
\end{proposition}
\begin{proof}
	By the law of total expectation,
	\[E[Y_k(t)]=E[E[Y_k(t)|\lambda]]=[\eta_k(t)+\eta_0(t)]E[\lambda^{-1}]=\frac{\rho}{\omega-1}\left[ \eta_k(t)+\eta_0(t)\right],
	\]
	\[\begin{aligned}
		\operatorname{Var}(Y_k(t)) &= E[\operatorname{Var}(Y_k(t) \mid \lambda)] + \operatorname{Var}(E[Y_k(t) \mid \lambda]) \\
		&=(\eta_k(t)+\eta_0(t)) E[\lambda^{-2}]+(\eta_k(t)+\eta_0(t))^2\operatorname{Var}[\lambda^{-1}]\\
		&=\frac{\rho^2(\eta_k(t)+\eta_0(t))}{(\omega-1)(\omega-2)}+\frac{\rho^2(\eta_k(t)+\eta_0(t))^2}{(\omega-1)^2(\omega-2)}\\
		&=\frac{\rho^2
			(\eta_k(t)+\eta_0(t))
			\left[
			\eta_k(t)+\eta_0(t)+\omega-1
			\right]}{(\omega-1)^2(\omega-2)}
	\end{aligned}
	\]
	The calculations of the inverse-order moments $E[\lambda^{-1}]$ and $E[\lambda^{-2}]$ are as follows
	\[\begin{aligned}
	E[\lambda^{-1}]&=\int_{0}^{\infty}\lambda^{-1}\frac{\rho^\omega}{\Gamma(\omega)}\lambda^{\omega-1}\exp(-\rho\lambda)\mathrm{d}\lambda\\
	&=\frac{\rho^\omega}{\Gamma(\omega)}\int_{0}^{\infty}\lambda^{\omega-2}\exp(-\rho\lambda)\mathrm{d}\lambda\\
	&=\frac{\rho}{\omega-1},
	\end{aligned}
	\]
	and
	\[\begin{aligned}
		E[\lambda^{-2}]&=\int_{0}^{\infty}\lambda^{-2}\frac{\rho^\omega}{\Gamma(\omega)}\lambda^{\omega-1}\exp(-\rho\lambda)\mathrm{d}\lambda\\
		&=\frac{\rho^\omega}{\Gamma(\omega)}\int_{0}^{\infty}\lambda^{\omega-3}\exp(-\rho\lambda)\mathrm{d}\lambda\\
		&=\frac{\rho^2}{(\omega-1)(\omega-2)}.
	\end{aligned}
	\]
	The covariance function between $Y_{k_1}(t)$ and $Y_{k_2}(t)$ can be derived as follows
	\[\begin{aligned}
	&\operatorname{Cov}[Y_{k_1}(t),Y_{k_2}(t)]\\
	&=E[Y_{k_1}(t)Y_{k_2}(t)]-E[Y_{k_1}(t)]E[Y_{k_2}(t)]\\
	&=E[E[Y_{k_1}(t)Y_{k_2}(t)]|\lambda]-E[Y_{k_1}(t)]E[Y_{k_2}(t)]\\
	&=E[E[(X_{k_1}(t)+Z(t))(X_{k_2}(t)+Z(t))]|\lambda]-E[Y_{k_1}(t)]E[Y_{k_2}(t)]\\
	&=E[E[X_{k_1}(t)X_{k_2}(t)]|\lambda]+E[E[X_{k_1}(t)Z(t)]|\lambda]+E[E[X_{k_2}(t)Z(t)]|\lambda]+E[E[Z^2(t)]|\lambda]\\
	&\quad-E[Y_{k_1}(t)]E[Y_{k_2}(t)]\\
	&=E[\frac{\eta_{k_1}(t)\eta_{k_2}(t)}{\lambda^2}]+E[\frac{\eta_{k_1}(t)\eta_{0}(t)}{\lambda^2}]+E[\frac{\eta_{k_2}(t)\eta_{0}(t)}{\lambda^2}]+E[\frac{\eta_{0}(t)(\eta_{0}(t)+1)}{\lambda^2}]\\
	&\quad-(\eta_{k_1}(t)+\eta_0(t))(\eta_{k_2}(t)+\eta_0(t))E^2[\lambda^{-1}]\\
	&=(\eta_{k_1}(t)+\eta_0(t))(\eta_{k_2}(t)+\eta_0(t))\{E[\lambda^{-2}]-E^2[\lambda^{-1}]\}+\eta_{0}(t)E[\lambda^{-2}]\\
	&=\frac{\rho^2(\eta_{k_1}(t)+\eta_0(t))(\eta_{k_2}(t)+\eta_0(t))}{(\omega-1)^2(\omega-2)}+\frac{\rho^2\eta_0(t)}{(\omega-1)(\omega-2)}.
	\end{aligned}
	\]
\end{proof}

Proposition \ref{pro1} demonstrates that the proposed model is capable of characterizing not only the marginal degradation behavior of individual performance characteristics (PCs) but also the dependence structure among different PCs. For any (k=1,\ldots,K), the mean and variance of the degradation process $Y_k(t)$ are jointly determined by the intrinsic degradation function $\eta_k(t)$) and the shock-induced degradation function $\eta_0(t)$). Since both functions are assumed to be monotonically increasing with time, the degradation level and its associated uncertainty increase progressively throughout the service life of the unit.

More importantly, the covariance expression derived in Proposition \ref{pro1} reveals that the dependence among different PCs originates from two distinct sources. Specifically,
\[
\operatorname{Cov}\{Y_{k_1}(t),Y_{k_2}(t)\}
=
\underbrace{
	\frac{\rho^2[\eta_{k_1}(t)+\eta_0(t)]
		[\eta_{k_2}(t)+\eta_0(t)]}
	{(\omega-1)^2(\omega-2)}
}_{\substack{\text{dependence induced}\\
		\text{by the shared random effect}}}
+
\underbrace{
	\frac{\rho^2\eta_0(t)}
	{(\omega-1)(\omega-2)}
}_{\substack{\text{dependence induced}\\
		\text{by the common shock process}}},
\quad k_1\neq k_2 .
\]

The covariance structure has a clear physical interpretation. Its first component is induced by the shared random effect factor $\lambda$, which characterizes unit-level heterogeneity due to common material, manufacturing, and operating conditions. Its second component arises from the common shock-induced degradation process $Z(t)$, through which cumulative external shock damage is shared by all PCs. Therefore, the proposed model captures two practically meaningful dependence sources: heterogeneity-induced dependence and shock-induced dependence, enabling an interpretable representation of synchronized degradation among multiple PCs.

Furthermore, the Gamma distribution is adopted to model the random-effect factor $\lambda$ for several reasons. First, the Gamma distribution ensures that the random-effect parameter remains positive, which is consistent with the physical constraints of degradation modeling. Second, the Gamma distribution is chosen because of its desirable conjugate properties, which facilitate the derivation of analytical results and simplify subsequent computation. This feature substantially simplifies both theoretical derivations and statistical inference. Finally, Gamma-distributed random effects have been widely adopted in the degradation modeling literature. For instance, \cite{lawless2004covariates},
\cite{tsai2012optimal}, and \cite{wang2021degradation} employed Gamma distributions to characterize unit-to-unit heterogeneity. Consequently, the Gamma distribution is
adopted to model the random-effect factor because it provides a
reasonable balance between physical interpretability, mathematical
tractability, and consistency with existing studies.
\begin{proposition}\label{pro2}
	The joint PDF of $Y_1(t),Y_2(t),\dots,Y_K(t)$ can be obtained as
	\begin{align}\label{a1}
		f_{\boldsymbol{Y}(t)}(y_1, \dots, y_K) = &\int_{0 }^{\infty}\int_{0}^{\tilde{y}}\prod_{k=1}^{K} f_{X_k(t)\mid\lambda}(y_k-z;\eta_k(t),\lambda)f_{Z(t)\mid\lambda}(z;\eta_0(t),\lambda)f_{\lambda}(\lambda;\omega,\rho)  \mathrm{d}z\mathrm{d}\lambda,\\\nonumber
	\end{align}
	where $\tilde{y} = \min\{y_1, \dots, y_K\}$, and $y_1, \dots, y_K$ are the observed degradation values. The corresponding CDF is expressed as
	\begin{align*}
		F_{\boldsymbol{Y}(t)}(y_1, \dots, y_K) = & \int_{0 }^{\infty}\int_{0}^{\tilde{y}}\prod_{k=1}^{K} F_{X_k(t)\mid\lambda}(y_k-z;\eta_k(t),\lambda)f_{Z(t)\mid\lambda}(z;\eta_0(t),\lambda)f_{\lambda}(\lambda;\omega,\rho)  \mathrm{d}z\mathrm{d}\lambda.
	\end{align*}
\end{proposition} 
Proposition \ref{pro2} gives the joint PDF and CDF of $Y_1(t), \dots, Y_K(t)$, which can be used to derive the ML estimates of the model parameters.
\section{Reliability analysis}\label{section3}
In engineering practice, the system lifetime is one of the most important quantities of interest in reliability analysis and maintenance planning. For degradation systems, failure is commonly defined through performance thresholds. Specifically, let $\tau_k$ denote the failure threshold associated with the $k$th PC. The $k$th PC is considered to fail when its degradation level first reaches or exceeds $\tau_k$. Furthermore, since the satisfactory operation of the system requires all PCs to remain within their acceptable performance ranges, the system is regarded as failed once any of the $K$ PCs reaches its corresponding failure threshold (\cite{lu2021general}; \cite{zhuang2024remaining}).

Based on the above failure criterion, the unit remains operational only if all performance characteristics stay below their corresponding failure thresholds. Therefore, the reliability function of the unit can be expressed as the joint probability that no PC reaches its threshold by time $t$. The following proposition provides the reliability function under the proposed model.

\begin{proposition}\label{pro3}
	The reliability function of the unit is given by
	\[
	R(t)
	=
	\int_0^\infty\int_0^{\tilde{\tau}}
	\prod_{k=1}^{K}
	F_{X_k(t)\mid\lambda}(\tau_k-z)
	f_{Z(t)\mid\lambda}(z)
	f_\lambda(\lambda;\rho,\omega)
	\,dz\,d\lambda,
	\]
	where
	\[
	\tilde{\tau}=\min\left\lbrace\tau_1,\tau_2,\dots,\tau_K \right\rbrace ,
	\]
	and
	\[
	F_{X_k(t)\mid\lambda}(\tau_k-z)
	=
	\int_0^{\tau_k-z}
	\frac{\lambda^{\eta_k(t)}}{\Gamma(\eta_k(t))}
	x_k^{\eta_k(t)-1}
	\exp(-\lambda x_k)\,dx_k .
	\]
\end{proposition}
\begin{proof}
	According to the failure criterion, the unit survives up to time $t$
	if and only if all PCs remain below their corresponding failure
	thresholds. Hence,
	\[ \begin{aligned} R(t)
	&=P(Y_1(t)<\tau_1,Y_2(t)<\tau_2,\dots,Y_K(t)<\tau_K)\\ &=P(X_1(t)+Z(t)<\tau_1,X_2(t)+Z(t)<\tau_2,\dots,X_K(t)+Z(t)<\tau_K)\\ &=\int_{0}^{\infty}P(X_1(t)+Z(t)<\tau_1,X_2(t)+Z(t)<\tau_2,\dots,X_K(t)+Z(t)<\tau_K|\lambda)f(\lambda;\omega,\rho)\mathrm{d}\lambda\\ &=\int_{0}^{\infty}\int_{0}^{\infty}P(X_1(t)<\tau_1-z,X_2(t)<\tau_2-z,\dots,X_K(t)<\tau_K-z|\lambda,Z(t)=z)f_{Z(t)|\lambda}(z)\\ 
	&\quad \times f(\lambda;\omega,\rho)\mathrm{d}z\mathrm{d}\lambda\\ &=\int_{0}^{\infty}\int_{0}^{\tilde{\tau}}P(X_1(t)<\tau_1-z,X_2(t)<\tau_2-z,\dots,X_K(t)<\tau_K-z|\lambda,Z(t)=z)f_{Z(t)|\lambda}(z)\\ 
	&\quad \times f(\lambda;\omega,\rho)\mathrm{d}z\mathrm{d}\lambda\\ &=\int_{0}^{\infty}\int_{0}^{\tilde{\tau}}\prod_{k=1}^{K}P(X_k(t)<\tau_k-z|\lambda)f_{Z(t)|\lambda}(z)f(\lambda;\omega,\rho)\mathrm{d}z\mathrm{d}\lambda\\ &=\int_{0}^{\infty}\int_{0}^{\tilde{\tau}}\prod_{k=1}^{K}F_{X_k(t)\mid\lambda}(\tau_k-z)f_{Z(t)|\lambda}(z)f(\lambda;\omega,\rho)\mathrm{d}z\mathrm{d}\lambda. \end{aligned}
	\]
		where
	\[
	\tilde{\tau}=\min\left\lbrace\tau_1,\tau_2,\dots,\tau_K \right\rbrace.
	\]
	The conditional distribution function $F_{X_k(t)\mid\lambda}(\tau_k-z)$ is given by
	\[
	\begin{aligned}
		F_{X_k(t)\mid\lambda}(\tau_k-z)=\frac{\gamma(\eta_k(t),\lambda(\tau_k-z))}{\Gamma(\eta_k(t))},
	\end{aligned}
	\]
	where $\gamma(s,x)=\int_{0}^{x}t^{s-1}e^{-t}~\mathrm{d} t$ denotes the lower incomplete Gamma function.
\end{proof}
	Since the integral in the reliability function is mathematically intractable to solve analytically, numerical approximation methods are required. In this paper, a hybrid Gauss quadrature method is employed to approximate this two-dimensional integral. Specifically, the Gauss-Legendre quadrature method is applied to the finite interval integral regarding the shared degradation level $z$. Meanwhile, the generalized Gauss-Laguerre quadrature method is utilized for the semi-infinite integral concerning the random effect factor $\lambda$. Compared with the conventional Monte Carlo integration, this deterministic quadrature approach offers superior stability and reproducibility, and it typically achieves comparable approximation accuracy with a significantly reduced computational budget.
	
    According to Proposition \ref{pro3}, the reliability function of the unit can be expressed as
	\begin{equation*}
		\begin{aligned}
			R(t) = {} & \frac{\rho^\omega}{\Gamma(\omega)\Gamma\{\eta_0(t)\}} \int_{0}^{\tilde{\tau}} z^{\eta_0(t)-1} \int_{0}^{\infty} \lambda^{\omega+\eta_0(t)-1} \exp\{-(\rho + z)\lambda\} \\
			& \times \prod_{k=1}^{K}F_{X_k(t)\mid\lambda}(\tau_k-z)  \, \mathrm{d}\lambda \, \mathrm{d}z\\
		    & = \frac{\rho^\omega}{\Gamma(\omega)\Gamma\{\eta_0(t)\}} \int_{0}^{\tilde{\tau}} \frac{z^{\eta_0(t)-1}}{(\rho + z)^{\omega+\eta_0(t)}} \\
			& \times \int_{0}^{\infty} s^{\omega+\eta_0(t)-1} \exp(-s) \prod_{k=1}^{K} P_{\Gamma} \left\{ \eta_k(t), \frac{s(\tau_k - z)}{\rho + z} \right\} \, \mathrm{d}s \, \mathrm{d}z.
		\end{aligned}
	\end{equation*}
	where $s=(\rho + z)\lambda$ and $P_\Gamma(a,x)$ denotes the regularized lower incomplete Gamma function with shape parameter $a$.
	
	For the outer integral defined over the finite interval $[0,\tilde{\tau}]$, we employ Gauss–Legendre quadrature. Let $u_q$ and $w_q$, $q=1,\ldots,l_z$, denote the nodes and weights of the $l_z$-point Gauss–Legendre rule on the standard interval $[-1,1]$. Applying the affine change of variables $z_q=\frac{\tilde{\tau}}{2}(u_q+1),\quad q=1,\ldots,l_z,$ which maps $[-1,1]$ onto $[0,\tilde{\tau}]$, the integral can be approximated by 
	\begin{eqnarray*}
		\int_{0}^{\tilde{\tau}} g(z)\,\mathrm{d}z \approx \frac{\tilde{\tau}}{2}\sum_{q=1}^{l_z} w_q\, g(z_q).
	\end{eqnarray*}
	
	For the inner integral defined over the interval $[0,\infty)$, let $a(t)=\omega+\eta_0(t)-1$. Then, the corresponding weight function for the generalized Gauss–Laguerre quadrature is $s^{a(t)}e^{-s}$. Let $s_r(t)$ and $W_r(t)$, $r=1,\ldots,l_\lambda$, denote the quadrature nodes and weights associated with this weight function, respectively. According to the generalized Gauss–Laguerre quadrature rule, the integral can be approximated as
	\begin{equation*}
		\int_{0}^{\infty} s^{\alpha(t)}e^{-s} h(s) \, \mathrm{d}s \approx \sum_{r=1}^{l_s} W_r h(s_r).
	\end{equation*}
	
	Combining the Gauss-Legendre approximation for $z$ and the generalized Gauss-Laguerre approximation for $s$, the reliability function can be approximated as
	\begin{equation}
		\begin{aligned}\label{a2}
			R(t) \approx {} & \frac{\rho^\omega}{\Gamma(\omega)\Gamma\{\eta_0(t)\}} \frac{\tilde{\tau}}{2} \sum_{q=1}^{l_z} w_q \frac{z_q^{\eta_0(t)-1}}{(\rho + z_q)^{\omega+\eta_0(t)}} \\
			& \times \sum_{r=1}^{l_s} W_r \prod_{k=1}^{K} P_{\Gamma} \left\{ \eta_k(t), \frac{s_r(\tau_k - z_q)}{\rho + z_q} \right\}.
		\end{aligned}
	\end{equation}

\section{Statistical inference}\label{section4}
    	This section considers the parameterized multivariate Gamma degradation model in which the shape functions take specific parametric forms. In the proposed additive degradation framework, $\eta_k(t)$ characterizes the accumulation pattern of the intrinsic degradation process of the $k$th PC, whereas $\eta_0(t)$ describes the cumulative effect of the common shock-induced degradation process. Therefore, the parametric specification of the shape functions determines not only the temporal evolution of the marginal degradation paths, but also the dependence structure among different PCs induced by the common shock process.
    	
    	For the $k$th PC, assume that the shape function of the intrinsic degradation process is given by $\eta_k(t)=\eta_k(t;\alpha_k,\beta_k)$, where $\alpha_k$ and $\beta_k$ are unknown parameters, $k=1,\ldots,K$. Commonly used specifications include the power-law form $\eta_k(t;\alpha_k,\beta_k)=\beta_k t^{\alpha_k}$ and the log-linear form $\eta_k(t;\alpha_k,\beta_k)=\beta_k{\exp(\alpha_k t-1)}$. In these formulations, $\alpha_k$ usually controls the nonlinear temporal evolution of the degradation process, while $\beta_k$ reflects the degradation accumulation level of the corresponding PC.
    	
    	However, under the additive model $Y_k(t)=X_k(t)+Z(t)$, the observed degradation is jointly driven by the intrinsic degradation process $X_k(t)$ and the common shock-induced degradation process $Z(t)$. Hence, the parametric forms of $\eta_k(t)$ and $\eta_0(t)$ should not be specified independently without additional restrictions. In particular, if the two components share the same temporal trend and both contain independent scale parameters, parameter confounding or weak identifiability may arise. For example, if $\eta_k(t)=\beta_k t^{\alpha}$ and $\eta_0(t)=\beta_0 t^{\alpha}$, then, conditional on $\lambda$, the marginal distribution of $Y_k(t)$ involves the accumulated shape effect only through $(\beta_k+\beta_0)t^{\alpha}$. In this case, the marginal degradation information mainly reflects the combined effect of $\beta_k+\beta_0$, making it difficult to stably distinguish the PC-specific degradation scale $\beta_k$ from the common shock-induced degradation scale $\beta_0$. Although the dependence structure in the multivariate joint distribution provides additional information about the common shock process, over-parameterization may still lead to unstable estimation when the sample size is limited and the temporal trends of the two components are highly similar.
  
    	To avoid such parameter redundancy and improve model interpretability, a more parsimonious specification is adopted for the common shock-induced degradation process by setting $\eta_0(t)=\eta_0(t;\alpha_0)$, without introducing an additional independent scale parameter $\beta_0$. This specification separates the common shock effect from the PC-specific intrinsic degradation scales: $\eta_0(t;\alpha_0)$ characterizes the accumulation pattern of common shock damage over time, whereas $\eta_k(t;\alpha_k,\beta_k)$ describes the PC-specific degradation rate and nonlinear temporal behavior. Consequently, the proposed specification reduces parameter confounding between the common degradation component and the marginal degradation components, while preserving the ability of the model to capture dependence among PCs with improved estimability and physical interpretability.
    	\subsection{Data and likelihood}\label{p1}
	   	Assume that $n$ units are tested in a degradation experiment. For the $i$-th unit, the degradation levels of $K$ performance characteristics (PCs) are measured at $m_i$ time points, denoted by $\boldsymbol{T}_i = (t_{i,1},\ldots,t_{i,m_i})^{\top}$, and $\boldsymbol{Y}_{i,k} = (Y_{i,k,1},\ldots,Y_{i,k,m_i})^{\top}$ for $i=1,\ldots,n,\quad k=1,2,\dots,K$.
	   For notational convenience, set $t_{i,0}=0$ and $Y_{i,k,0}=0$. The degradation increment of the $k$th PC over the interval $(t_{i,j-1},t_{i,j}]$ is defined as
	   $\Delta Y_{i,k,j} = Y_{i,k,j}-Y_{i,k,j-1}, \quad j=1,\ldots,m_i$.
	   	Furthermore, let $\Delta \boldsymbol{Y}_{i,:,j} = (\Delta Y_{i,1,j},\ldots,\Delta Y_{i,K,j})^{\top}$
	   	denote the vector of degradation increments over the $j$th interval for the $i$-th unit. The model parameter vector	is defined as $\boldsymbol{\Theta}=\{\alpha_k,\beta_k,\alpha_0,\rho,\omega\}^{\top}$.
	   	Based on (\ref{a1}), conditional on the unit-specific realization $\lambda_i$ of the shared random-effect factor, the density of the increment vector $\Delta \boldsymbol{Y}_{i,:,j}$ can be expressed as
	   	\begin{equation*}
	   		p(\Delta \boldsymbol{y}_{i,:,j}\mid \lambda_i;\boldsymbol{\Theta}) = \int_{0}^{\Delta \tilde{y}_{i,j}} f_G(\Delta z_{i,j};\lambda_i,\Delta\eta_{i,0,j}) \prod_{k=1}^{K} f_G(\Delta y_{i,k,j}-\Delta z_{i,j};\lambda_i,\Delta\eta_{i,k,j}) \,\mathrm{d}\Delta z_{i,j},
	   	\end{equation*}
	   	where
	   	\begin{equation*}
	   		\Delta \tilde{y}_{i,j} = \min\{\Delta y_{i,1,j},\Delta y_{i,2,j},\ldots,\Delta y_{i,K,j}\}.
	   	\end{equation*}
	   	Here, $f_G(x;\lambda,a)$ denotes the Gamma density with rate parameter $\lambda$ and shape parameter $a$, namely
	   	\begin{equation*}
	   		f_G(x;\lambda,a) = \frac{\lambda^{a}}{\Gamma(a)} x^{a-1}\exp(-\lambda x), \qquad x>0.
	   	\end{equation*}
	   	The shape increment of the intrinsic degradation process for the $k$th PC over $(t_{i,j-1},t_{i,j}]$ is defined as
	   	\begin{equation*}
	   		\Delta\eta_{i,k,j} = \eta_k(t_{i,j};\alpha_k,\beta_k) - \eta_k(t_{i,j-1};\alpha_k,\beta_k), \qquad k=1,\ldots,K,
	   	\end{equation*}
	   	and the shape increment of the common shock-induced degradation process is defined as
	   	\begin{equation*}
	   		\Delta\eta_{i,0,j} = \eta_0(t_{i,j};\alpha_0) - \eta_0(t_{i,j-1};\alpha_0).
	   	\end{equation*}
	   	
	   	Since the same random-effect factor $\lambda_i$ is shared by all degradation increments of the $i$-th unit, the likelihood contribution of this unit is
	   	\begin{equation*}
	   		L_i(\boldsymbol{\Theta}) = \int_{0}^{\infty} \left[ \prod_{j=1}^{m_i} p(\Delta \boldsymbol{y}_{i,:,j}\mid \lambda_i;\boldsymbol{\Theta}) \right] f_{\lambda}(\lambda_i;\omega,\rho) \,\mathrm{d}\lambda_i .
	   	\end{equation*}
	   	Assuming that different units are independent, the full likelihood function for the observed degradation data is
	   	\begin{equation*}
	   		L(\boldsymbol{\Theta}) = \prod_{i=1}^{n} L_i(\boldsymbol{\Theta}).
	   	\end{equation*}
	   	Equivalently,
	   	\begin{equation*}
	   		L(\boldsymbol{\Theta}) = \prod_{i=1}^{n} \int_{0}^{\infty} \left[ \prod_{j=1}^{m_i} \int_{0}^{\Delta \tilde{y}_{i,j}} f_G(\Delta z_{i,j};\lambda_i,\Delta\eta_{i,0,j}) \prod_{k=1}^{K} f_G(\Delta y_{i,k,j}-\Delta z_{i,j};\lambda_i,\Delta\eta_{i,k,j}) \,\mathrm{d}\Delta z_{i,j} \right] f_{\lambda}(\lambda_i;\omega,\rho) \,\mathrm{d}\lambda_i .
	   	\end{equation*}
	   	The corresponding log-likelihood function is
	   	\begin{equation*}
	   		\ell(\boldsymbol{\Theta}) = \sum_{i=1}^{n} \log \left\{ \int_{0}^{\infty} \left[ \prod_{j=1}^{m_i} p(\Delta \boldsymbol{y}_{i,:,j}\mid \lambda_i;\boldsymbol{\Theta}) \right] f_{\lambda}(\lambda_i;\omega,\rho) \,\mathrm{d}\lambda_i \right\}.
	   	\end{equation*}
	   	The maximum likelihood estimator is obtained by
	   	\begin{equation*}
	   		\hat{\boldsymbol{\Theta}} = \arg\max_{\boldsymbol{\Theta}} \ell(\boldsymbol{\theta}).
	   	\end{equation*}
	   	However, the log-likelihood function involves nested integrals and is a complicated nonlinear function of $\boldsymbol{\Theta}$. Consequently, a closed-form solution for $\hat{\boldsymbol{\Theta}}$ is generally unavailable. To address this difficulty, an EM algorithm is developed to obtain the maximum likelihood estimates.
	\subsection{Point estimation by the EM algorithm}\label{p2}
	For the EM algorithm, the unobserved common shock-induced degradation ($Z_{i,j}=Z_i(t_{i,j})$) and the shared random effects ($\lambda_i$), ($i=1,\ldots,n$), ($j=1,\ldots,m_i$), are treated as latent variables. Define $\boldsymbol{Z}_i=(Z_{i,1},Z_{i,2},\ldots,Z_{i,m_i})^{\top}$.
		For notational convenience, set $Z_{i,0}=0$. Conditional on $\lambda_i$, the increment of the common shock-induced degradation process over the interval ($(t_{i,j-1},t_{i,j}]$) is $ \Delta Z_{i,j}=Z_{i,j}-Z_{i,j-1}$,
		and follows a Gamma distribution: $\Delta Z_{i,j}\mid \lambda_i \sim Ga(\Delta \eta_{i,0,j},\lambda_i)$.
		Furthermore, given $\boldsymbol{Z}_i$ and $\lambda_i$, the intrinsic degradation increment of the $k$-th PC satisfies $\Delta Y_{i,k,j}-\Delta Z_{i,j} \mid \boldsymbol{Z}_i,\lambda_i \sim Ga(\Delta \eta_{i,k,j},\lambda_i)$,
		with$0\leq \Delta Z_{i,j}\leq \Delta \tilde y_{i,j}$.
		Let $\mathbb{Y} = \{\boldsymbol{Y}_1,\boldsymbol{Y}_2,\dots,\boldsymbol{Y}_n\}$
		denote the observed degradation data,
		and let $\mathbb{Z} = \{\boldsymbol{Z}_1,\ldots,\boldsymbol{Z}_n\}$
		denote the latent common shock-induced degradation. 
		
	Then for given $\boldsymbol{Z}_i$ and $\lambda_i$, the conditional likelihood function is 
	\begin{align}\label{a18}
		\mathcal{L}(\boldsymbol{Y}_i\mid\boldsymbol{Z}_i,\lambda_i)&=\prod_{j=1}^{m_i}\prod_{k=1}^{K}f_G(\Delta y_{i,k,j}-\Delta z_{i,j};\Delta\eta_{i,k,j},\lambda_i)\nonumber\\
		&=\prod_{j=1}^{m_i}\prod_{k=1}^{K}\frac{\lambda_i^{\Delta\eta_{i,k,j}}}{\Gamma(\Delta\eta_{i,k,j})}\exp[-\lambda_i(\Delta y_{i,k,j}-\Delta z_{i,j})](y_{i,k,j}-\Delta z_{i,j})^{\eta_{i,k,j}-1}.
	\end{align}
	The PDF of $\boldsymbol{Z}_i$ conditional an $\lambda_i$ is given by
	\begin{align}\label{a19}
		\Psi(\boldsymbol{Z}\mid\lambda_i)&=\prod_{j=1}^{m_i}f_G(\Delta z_{i,j};\Delta\eta_{i,0,j},\lambda_i)\nonumber\\
		&=\prod_{j=1}^{m_i}\frac{\lambda_i^{\Delta\eta_{i,0,j}}}{\Gamma(\Delta\eta_{i,0,j})}\exp[-\lambda_i\Delta z_{i,j}]\Delta z_{i,j}^{\eta_{i,k,j}-1}.
	\end{align}
	The unconditional PDF of $\lambda_i$ is 
	\begin{align}\label{a20}
		\varphi(\lambda_i)=f_G(\lambda_i;\omega,\rho)=\frac{\rho^\omega}{\Gamma(\omega)}\exp(-\rho\lambda_i)\lambda_i^{\omega-1}.
	\end{align}
	Based on the complete data $\{\mathbb{Y},\mathbb{Z},\boldsymbol{\lambda}\}$, by combining \eqref{a18}, \eqref{a19} and \eqref{a20}, the complete data log-likelihood function is given by
	\begin{equation*}
		\begin{aligned}
			\ell_c(\boldsymbol{\Theta};\mathbb{Y},\mathbb{Z},\boldsymbol{\lambda}) &= \sum_{i=1}^{n} \bigg\{ \log f_{G}(\lambda_i;\omega,\rho) \\
			&\quad + \sum_{j=1}^{m_i} \bigg[ \log f_G(\Delta z_{i,j};\Delta\eta_{i,0,j},\lambda_i) + \sum_{k=1}^{K} \log f_G(\Delta y_{i,k,j}-\Delta z_{i,j}; \Delta\eta_{i,k,j},\lambda_i) \bigg] \bigg\}.
		\end{aligned}
	\end{equation*}
	
	The EM algorithm iteratively alternates between the E-step and the M-step to obtain the maximum likelihood estimates of the model parameters. Let $\boldsymbol{\Theta}^{(r)}$ denote the parameter estimate at the $r$th iteration. In the $(r+1)$th iteration, the E-step first computes the conditional expectation of the complete data log-likelihood, known as the $Q$-function:
	\begin{align}\label{Q-f}
		Q(\boldsymbol{\Theta} \mid \boldsymbol{\Theta}^{(r)}) = \mathbb{E} \left[ \ell_c(\boldsymbol{\Theta}; \mathbb{Y}, \mathbb{Z}, \boldsymbol{\lambda}) \mid \mathbb{Y}, \boldsymbol{\Theta}^{(r)}\right].
	\end{align}
	This expectation is taken with respect to the posterior distribution of the latent variables $(\mathbb{Z}, \boldsymbol{\lambda})$ given the observed data $\mathbb{Y}$ and the current parameter estimate $\boldsymbol{\Theta}^{(r)}$.
	
	To proceed the E-step, we need to calculate
	\begin{align}\label{a12}
		&\zeta_i=\mathbb{E}[\log(\lambda_i)\mid \boldsymbol{Y}_i,\boldsymbol{\Theta}^{(r)}],~\xi_{i}=\mathbb{E}[\lambda_i\mid \boldsymbol{Y}_i,\boldsymbol{\Theta}^{(r)}],~\pi_{i,j}=\mathbb{E}[\lambda_i\Delta Z_{i,j}\mid \boldsymbol{Y}_i,\boldsymbol{\Theta}^{(r)}],\nonumber\\
		&\sigma_{i,k,j}=\mathbb{E}[\log(\Delta y_{i,k,j}-\Delta Z_{i,j})\mid \boldsymbol{Y}_i,\boldsymbol{\Theta}^{(r)}],~ \textbf{and} ~\nu_{i,j}=\mathbb{E}[\log(\Delta Z_{i,j})\mid \boldsymbol{Y}_i,\boldsymbol{\Theta}^{(r)}].
	\end{align}
	
	When we obtain the $Q$-function in (\ref{Q-f}), the update of the optimal solution in the M-step is carried out as follows:
	\begin{align}\label{a9}
		\boldsymbol{\Theta}^{(r+1)}=\arg\max_{\boldsymbol{\Theta}}Q(\boldsymbol{\Theta} \mid \boldsymbol{\Theta}^{(r)}).
	\end{align}
	The M-step updates $\boldsymbol{\Theta}$ by maximizing $Q(\boldsymbol{\Theta}\mid \boldsymbol{\Theta}^{(r)})$. Differentiating the $Q$-function in (\ref{Q-f}) with respect to the components of $\boldsymbol{\Theta}$ and equating the derivatives to zero gives the following estimating equations:
	\begin{align}
		&\frac{\partial Q(\boldsymbol{\Theta}\mid\boldsymbol{\Theta}^{(r)})}{\partial \alpha_k} = \sum_{i=1}^{n} \sum_{j=1}^{m_i} \frac{\partial \Delta \eta_{i,k,j}}{\partial \alpha_k} \left\{ E[\log\lambda_i\mid\boldsymbol{Y}_i,\boldsymbol{\Theta}^{(r)}]-\psi(\Delta\eta_{i,k,j})+E[\log(\Delta Y_{i,k,j}-\Delta Z_{i,j})\mid\boldsymbol{Y}_i,\boldsymbol{\Theta}^{(r)}]\right\} = 0,\label{a6}\\
		&\frac{\partial Q(\boldsymbol{\Theta}\mid\boldsymbol{\Theta}^{(r)})}{\partial \beta_k} = \sum_{i=1}^{n} \sum_{j=1}^{m_i} \frac{\partial \Delta \eta_{i,k,j}}{\partial \beta_k} \left\{ E[\log\lambda_i\mid\boldsymbol{Y}_i,\boldsymbol{\Theta}^{(r)}]-\psi(\Delta\eta_{i,k,j})+E[\log(\Delta Y_{i,k,j}-\Delta Z_{i,j})\mid\boldsymbol{Y}_i,\boldsymbol{\Theta}^{(r)}]\right\} = 0,\label{a7}\\
		&\frac{\partial Q(\boldsymbol{\Theta}\mid\boldsymbol{\Theta}^{(r)})}{\partial \alpha_0} = \sum_{i=1}^{n} \sum_{j=1}^{m_i} \frac{\partial \Delta \eta_{i,0,j}}{\partial \alpha_0} \left\{ E[\log\lambda_i\mid\boldsymbol{Y}_i,\boldsymbol{\Theta}^{(r)}]-\psi(\Delta\eta_{i,0,j})+E[\log(\Delta Z_{i,j})\mid\boldsymbol{Y}_i,\boldsymbol{\Theta}^{(r)}]\right\} = 0,\label{a8}\\
		&\frac{\partial Q(\boldsymbol{\Theta}\mid\boldsymbol{\Theta}^{(r)})}{\partial \rho}=\sum_{i=1}^{n}\{\frac{\omega}{\rho}-E[\lambda_i\mid\boldsymbol{Y}_i,\boldsymbol{\Theta}^{(r)}]\}=0,\label{a3}\\
		&\frac{\partial Q(\boldsymbol{\Theta}\mid\boldsymbol{\Theta}^{(r)})}{\partial \omega}=\sum_{i=1}^{n}\{\log\rho-\psi(\omega)+E[\log\lambda_i\mid\boldsymbol{Y}_i,\boldsymbol{\Theta}^{(r)}]\}=0,\label{a4}
	\end{align}
	for $k=1,2,\dots,K$, where $\psi(\cdot)$ denotes the digamma function. 	From the estimating equation for $\rho$, we obtain
	\begin{equation*}
		\rho^{(r+1)} = \frac{n\omega^{(r+1)}}{\sum_{i=1}^{n} \mathbb{E}[\lambda_i\mid\boldsymbol{Y}_i,\boldsymbol{\Theta}^{(r)}]}.
	\end{equation*}
	Substituting this expression into the estimating equation for $\omega$ yields
	\begin{equation}\label{a5}
		\log\omega^{(r+1)} - \psi(\omega^{(r+1)}) = \log\left\{ \frac{1}{n} \sum_{i=1}^{n} \mathbb{E}[\lambda_i\mid\boldsymbol{Y}_i,\boldsymbol{\Theta}^{(r)}] \right\} - \frac{1}{n} \sum_{i=1}^{n} \mathbb{E}[\log\lambda_i\mid\boldsymbol{Y}_i,\boldsymbol{\Theta}^{(r)}].
	\end{equation}
	\cite{zaigraev2008estimation} showed that $\log(\omega)-\psi(\omega)$ is a strictly decreasing function ranging from $0$ to $\infty$, and we can
	verify $\log({1}/{n}\sum_{i=1}^{n}E[\lambda_i\mid\boldsymbol{Y}_i,\boldsymbol{\Theta}^{(r)}])-{1}/{n}\sum_{i=1}^{n}E[\log\lambda_i\mid\boldsymbol{Y}_i,\boldsymbol{\Theta}^{(r)}]>0$ via the Jensen's inequality. Thus the root of (\ref{a5}) is well-defined and unique.
	
	Here,  The expectations are taken with respect to the conditional distribution of the latent variables given $\boldsymbol{Y}_i$ and $\boldsymbol{\Theta}^{(r)}$, and are therefore treated as constants in the M-step.
	
	From the estimating equation for $\rho$, we obtain
	\begin{equation}
		\rho^{(r+1)} = \frac{n\omega^{(r+1)}}{\sum_{i=1}^{n} \mathbb{E}[\lambda_i\mid\boldsymbol{Y}_i,\boldsymbol{\Theta}^{(r)}]}.
	\end{equation}
	Substituting this expression into the estimating equation for $\omega$ yields
	\begin{equation}
		\log\omega^{(r+1)} - \psi(\omega^{(r+1)}) = \log\left\{ \frac{1}{n} \sum_{i=1}^{n} \mathbb{E}[\lambda_i\mid\boldsymbol{Y}_i,\boldsymbol{\Theta}^{(r)}] \right\} - \frac{1}{n} \sum_{i=1}^{n} \mathbb{E}[\log\lambda_i\mid\boldsymbol{Y}_i,\boldsymbol{\Theta}^{(r)}].
	\end{equation}
	The equation above can be solved numerically for $\omega^{(r+1)}$. Since $g(\omega)=\log\omega-\psi(\omega)$ is strictly decreasing on $(0,\infty)$, with $g(\omega)\to\infty$ as $\omega\to0^+$ and $g(\omega)\to0$ as $\omega\to\infty$, the solution is well defined and unique whenever the right-hand side is positive. This positivity follows from Jensen's inequality under the usual nondegenerate posterior distribution of $\lambda_i$. 
	
   In addition, the updates of $\alpha_k$, $\beta_k$, and $\alpha_0$ can be obtained by solving the estimating equations in \eqref{a6}-\eqref{a8}, where $k=1,2,\ldots,K$.
   
   	Given initial values $\boldsymbol{\Theta}^{(0)}$, the EM algorithm iterates until a predefined convergence criterion is satisfied. Specifically, convergence is achieved when $\left\|\boldsymbol{\Theta}^{(r+1)} - \boldsymbol{\Theta}^{(r)}\right\| < \epsilon_1$, where $\left\|\cdot\right\|$ represents the absolute difference and $\epsilon$ is a preset tolerance. The converged value is taken as the maximum likelihood estimate of $\boldsymbol{\Theta}$. The implementation of the proposed EM algorithm is summarized in Algorithm~\ref{alg:EM}.
%
%
%
%
%
   	\begin{algorithm}
   		\caption{Implementation of the proposed EM algorithm.}
   		\label{alg:EM}
   		\begin{algorithmic}
   			\Statex \textbf{Input:} $\mathbb{Y}$, $\boldsymbol{\Theta}^{(0)}$, and  $\epsilon_1$.
   			\Statex \textbf{Output:} 
   			$\widehat{\boldsymbol{\Theta}}
   			=
   			(\hat{\alpha}_k,\hat{\beta}_k,\hat{\alpha}_0,\hat{\rho},\hat{\omega})^{\top}$,~k=1,2,\dots,K.
   			\State Set $r=0$.
   			\Repeat
   			\Statex \textbf{E-step:}
   			\State Compute the required conditional expectations
   			\[
   			\mathbb{E}[\log(\lambda_i)\mid \boldsymbol{Y}_i,\boldsymbol{\Theta}^{(r)}],
   			\quad
   			\mathbb{E}[\lambda_i\mid \boldsymbol{Y}_i,\boldsymbol{\Theta}^{(r)}],
   			\quad
   			\mathbb{E}[\lambda_i\Delta Z_{i,j}\mid \boldsymbol{Y}_i,\boldsymbol{\Theta}^{(r)}],
   			\]
   			\[
   			\mathbb{E}[\log(\Delta y_{i,k,j}-\Delta Z_{i,j})\mid \boldsymbol{Y}_i,\boldsymbol{\Theta}^{(r)}],
   			\quad
   			\mathbb{E}[\log(\Delta Z_{i,j})\mid \boldsymbol{Y}_i,\boldsymbol{\Theta}^{(r)}].
   			\]
   			\Statex \textbf{M-step:}
   			\State Update $\boldsymbol{\Theta}^{(r+1)}$ by solving the estimating equations in \eqref{a6}-\eqref{a5}.
   			\State Set $r=r+1$.
   			\Until{$\left\|\boldsymbol{\Theta}^{(r)}-\boldsymbol{\Theta}^{(r-1)}\right\|<\epsilon_1$}
   			\State Return $\widehat{\boldsymbol{\Theta}}=\boldsymbol{\Theta}^{(r)}$.
   		\end{algorithmic}
   	\end{algorithm}
   	\subsection{Determine initial parameter estimation values}\label{p3}
   	The choice of initial parameter values is important for implementing the EM algorithm, as it can affect both the convergence speed and the numerical stability of the final estimates. To obtain reasonable initial values, a stepwise initialization strategy is adopted in this study. First, the shared random effect is temporarily treated as a fixed but unknown parameter and is preliminarily estimated from the corresponding degradation data. Then, the hyperparameters of the shared random-effect distribution are estimated using the method of moments. Finally, the parameters $\alpha$ and $\beta$ in the degradation shape functions are initialized by minimizing a weighted mean squared error criterion.
   	
   	Specifically, we first compute the sample mean and sample variance of the degradation increments for the $i$th unit:
   	\begin{align*}
   		&\Delta \bar{y}_i=\frac{1}{Km_i}\sum_{k=1}^{K}\sum_{j=1}^{m_i}\Delta y_{i,k,j}, ~\text{and}\\
   		&\Delta s_i^2=\frac{1}{Km_i-1}\sum_{k=1}^{K}\sum_{j=1}^{m_i}(\Delta y_{i,k,j}-\Delta \bar{y}_i)^2.
   	\end{align*}
   	
   	According to \ref{a10} and the principles of moment estimation, we calculate the estimate for $\hat{\lambda}_i,~i=1,2,\dots,n$
   	\begin{align*}
   		\hat{\lambda}_i=\frac{\Delta \bar{y}_i}{\Delta s_i^2}.
   	\end{align*}
   	Based on the preliminary estimates $\hat{\lambda}_i$, $i=1,\ldots,n$, obtained from the $n$ units, we fit a Gamma distribution $Ga(\omega,\rho)$ to $\{\hat{\lambda}_i,~i=1,2,\dots,n\}$ and obtain the initial hyperparameter values $\omega^{(0)}$ and $\rho^{(0)}$ as
   	\begin{equation*}
   		\omega^{(0)} = \frac{\bar{\lambda}^{\,2}}{s_{\hat{\lambda}}^2}, \qquad \rho^{(0)} = \frac{\bar{\lambda}}{s_{\hat{\lambda}}^2},
   	\end{equation*}
   	where
   	\begin{equation*}
   		\bar{\lambda} = \frac{1}{n}\sum_{i=1}^{n}\hat{\lambda}_i, \qquad s_{\hat{\lambda}}^2 = \frac{1}{n-1} \sum_{i=1}^{n} \left( \hat{\lambda}_i-\bar{\lambda} \right)^2 .
   	\end{equation*}
   	
   	Since $E[\Delta Y_{i,k,j}]$ represents the mean degradation increment of the $k$-th PC for the $i$-th unit over the observation interval $(t_{i,j-1},t_{i,j}]$, the initial values of the parameters in the degradation shape functions can be obtained by minimizing the weighted mean squared error between $E[\Delta Y_{i,k,j}]$ and the observed degradation increment $\Delta y_{i,k,j}$. Given the initial estimates $\rho^{(0)}$ and $\omega^{(0)}$, the initial estimate of
   	\begin{equation*}
   		\boldsymbol{\Psi} = (\alpha_1,\beta_1,\ldots,\alpha_K,\beta_K,\alpha_0)^{\top}
   	\end{equation*}
   	is obtained by solving the following optimization problem:
   	\begin{equation*}
   		\boldsymbol{\Psi}^{(0)} = \arg\min_{\boldsymbol{\Psi}} \sum_{i=1}^{n} \sum_{j=1}^{m_i} \sum_{k=1}^{K} \frac{ \left\{ E[\Delta Y_{i,k,j}] - \Delta y_{i,k,j} \right\}^2 }{ \operatorname{Var}[\Delta Y_{i,k,j}] } .
   	\end{equation*}
   	\subsection{Variational inference for computing conditional expectations}\label{p4}
   	
   	In the E-step, the conditional expectations required by the EM algorithm are computed with respect to the posterior distribution of the latent variables $(\boldsymbol{Z}_i,{\lambda}_i)$, namely
   	\begin{equation*}
   		p(\boldsymbol{Z}_i,{\lambda}_i\mid \boldsymbol{Y}_i;\boldsymbol{\Theta}^{(r)}).
   	\end{equation*}
   	For notational simplicity, the dependence on the current parameter estimate $\boldsymbol{\Theta}^{(r)}$ is omitted hereafter. By Bayes' theorem, the joint density of the observed data, the latent common shock-induced degradation paths, and the shared random effects can be factorized as
   	\begin{equation*}
   		\mathbb{W}(\boldsymbol{Y}_i,\boldsymbol{Z}_i,{\lambda}_i) = \mathcal{L}(\boldsymbol{Y}_i\mid\boldsymbol{Z}_i,\lambda_i)\Psi(\boldsymbol{Z}_i\mid {\lambda}_i) \varphi({\lambda}_i).
   	\end{equation*}
   	Therefore, the joint posterior distribution of $\boldsymbol{Z}_i$ and ${\lambda}_i$ given $\boldsymbol{Y}_i$ is given by
   	\begin{equation}\label{a11}
   		\mathbb{W}(\boldsymbol{Z}_i,{\lambda}_i\mid\boldsymbol{Y}_i) = \frac{ \mathbb{W}(\boldsymbol{Y}_i,\boldsymbol{Z}_i,{\lambda}_i) }{ \int\!\!\int \mathbb{W}(\boldsymbol{Y}_i,\boldsymbol{Z}_i,\boldsymbol{\lambda}) \,\mathrm{d}\boldsymbol{Z}_i\,\mathrm{d}{\lambda}_i } .
   	\end{equation}
   	This posterior distribution is analytically intractable, and hence the conditional expectations in \eqref{a12} are not available in closed form.
   	
   	To address the analytical intractability of the posterior distribution, variational inference is adopted to approximate the true posterior distribution $\mathbb{W}(\boldsymbol{Z}_i,{\lambda}_i\mid\boldsymbol{Y}_i)$ by a factorized variational distribution, that is
   		\begin{equation}
   			\mathbb{W}(\boldsymbol{Z}_i,{\lambda}_i\mid\boldsymbol{Y}_i) \approx J(\boldsymbol{Z}_i\mid\boldsymbol{Y}_i)q({\lambda}_i\mid\boldsymbol{Y}_i).
   		\end{equation}
   		The optimal variational densities $J^*(\boldsymbol{Z}_i\mid\boldsymbol{Y}_i)$ and $q^*({\lambda}_i\mid\boldsymbol{Y}_i)$ are then obtained by minimizing the following Kullback-Leibler divergence
   		\begin{align}\label{a13}
   			D_{\mathrm{KL}}(J(\boldsymbol{Z}_i\mid\boldsymbol{Y}_i)q({\lambda}_i\mid\boldsymbol{Y}_i)\,\big\|\,\mathbb{W}(\boldsymbol{Z}_i,{\lambda}_i\mid\boldsymbol{Y}_i))=\int\!\!\int J(\boldsymbol{Z}_i\mid\boldsymbol{Y}_i)q({\lambda}_i\mid\boldsymbol{Y}_i)\log(\frac{J(\boldsymbol{Z}_i\mid\boldsymbol{Y}_i)q({\lambda}_i\mid\boldsymbol{Y}_i)}{\mathbb{W}(\boldsymbol{Z}_i,{\lambda}_i\mid\boldsymbol{Y}_i)})\,\mathrm{d}\boldsymbol{Z}_i\mathrm{d}{\lambda}_i.
   		\end{align}
   		By combining \eqref{a11} and \eqref{a13}, we have
   		\begin{align*}
   			&\quad\log(\int\!\!\int \mathbb{W}(\boldsymbol{Z}_i,{\lambda}_i,\boldsymbol{Y}_i)\,\mathrm{d}\boldsymbol{Z}_i\mathrm{d}{\lambda}_i)\\
   			&=\log\frac{\mathbb{W}(\boldsymbol{Z}_i,{\lambda}_i,\boldsymbol{Y}_i)}{	\mathbb{W}(\boldsymbol{Z}_i,{\lambda}_i\mid\boldsymbol{Y}_i)}\\
   			&=\int\!\!\int J(\boldsymbol{Z}_i\mid\boldsymbol{Y}_i)q({\lambda}_i\mid\boldsymbol{Y}_i)[\log \mathbb{W}(\boldsymbol{Z}_i,{\lambda}_i,\boldsymbol{Y}_i)-\log \mathbb{W}(\boldsymbol{Z}_i,{\lambda}_i\mid\boldsymbol{Y}_i)]\,\mathrm{d}\boldsymbol{Z}_i\mathrm{d}{\lambda}_i\\
   			&=\int\!\!\int J(\boldsymbol{Z}_i\mid\boldsymbol{Y}_i)q({\lambda}_i\mid\boldsymbol{Y}_i)\log \mathbb{W}(\boldsymbol{Z}_i,{\lambda}_i,\boldsymbol{Y}_i)\,\mathrm{d}\boldsymbol{Z}_i\mathrm{d}{\lambda}_i-\int\!\!\int J(\boldsymbol{Z}_i\mid\boldsymbol{Y}_i)q({\lambda}_i\mid\boldsymbol{Y}_i)\log \mathbb{W}(\boldsymbol{Z}_i,{\lambda}_i\mid\boldsymbol{Y}_i)\,\mathrm{d}\boldsymbol{Z}_i\mathrm{d}{\lambda}_i\\
   			&=\int\!\!\int J(\boldsymbol{Z}_i\mid\boldsymbol{Y}_i)q({\lambda}_i\mid\boldsymbol{Y}_i)\log \mathbb{W}(\boldsymbol{Z}_i,{\lambda}_i,\boldsymbol{Y}_i)\,\mathrm{d}\boldsymbol{Z}_i\mathrm{d}{\lambda}_i-D_{\mathrm{KL}}(J(\boldsymbol{Z}_i\mid\boldsymbol{Y}_i)q({\lambda}_i\mid\boldsymbol{Y}_i)\,\big\|\,\mathbb{W}(\boldsymbol{Z}_i,{\lambda}_i,\boldsymbol{Y}_i))\\
   			&\quad-\left\lbrace \int\!\!\int J(\boldsymbol{Z}_i\mid\boldsymbol{Y}_i)q({\lambda}_i\mid\boldsymbol{Y}_i)\log J(\boldsymbol{Z}_i\mid\boldsymbol{Y}_i)q({\lambda}_i\mid\boldsymbol{Y}_i)\,\mathrm{d}\boldsymbol{Z}_i\mathrm{d}{\lambda}_i-D_{\mathrm{KL}}(J(\boldsymbol{Z}_i\mid\boldsymbol{Y}_i)q({\lambda}_i\mid\boldsymbol{Y}_i)\,\big\|\,\mathbb{W}(\boldsymbol{Z}_i,{\lambda}_i\mid\boldsymbol{Y}_i))\right\rbrace \\
   			&=D_{\mathrm{KL}}(J(\boldsymbol{Z}_i\mid\boldsymbol{Y}_i)q({\lambda}_i\mid\boldsymbol{Y}_i)\,\big\|\,\mathbb{W}(\boldsymbol{Z}_i,{\lambda}_i\mid\boldsymbol{Y}_i))+\mathcal{L}(J,q),
   		\end{align*}
   		where $\mathcal{L}(J,q)=\int\!\!\int J(\boldsymbol{Z}_i\mid\boldsymbol{Y}_i)q({\lambda}_i\mid\boldsymbol{Y}_i)\log\frac{ \mathbb{W}(\boldsymbol{Z}_i,{\lambda}_i,\boldsymbol{Y}_i)}{J(\boldsymbol{Z}_i\mid\boldsymbol{Y}_i)q({\lambda}_i\mid\boldsymbol{Y}_i)}\,\mathrm{d}\boldsymbol{Z}_i\mathrm{d}{\lambda}_i$. 
   		Noting that $
   				\log\left( \int\!\!\int \mathbb{W}(\boldsymbol{Z}_i,{\lambda}_i,\boldsymbol{Y}_i) \,\mathrm{d}\boldsymbol{Z}_i\mathrm{d}{\lambda}_i\right)$
   			depends only on the observed data $\boldsymbol{Y}_i$ and is independent of the variational densities $J(\boldsymbol{Z}_i\mid\boldsymbol{Y}_i)$ and $q({\lambda}_i\mid\boldsymbol{Y}_i)$, it can be treated as a constant. Therefore, minimizing the Kullback--Leibler divergence in \eqref{a13} is equivalent to maximizing the evidence lower bound $\mathcal{L}(J,q)$. Further, $q({\lambda}_i\mid\boldsymbol{Y}_i)$ is fixed, the objective function can be written as
   			\begin{align*}
   				\mathcal{L}(J,q)&=\int\!\!\int J(\boldsymbol{Z}_i\mid\boldsymbol{Y}_i)q({\lambda}_i\mid\boldsymbol{Y}_i)[\log \mathbb{W}(\boldsymbol{Z}_i,{\lambda}_i,\boldsymbol{Y}_i)-\log J(\boldsymbol{Z}_i\mid\boldsymbol{Y}_i)-\log q({\lambda}_i\mid\boldsymbol{Y}_i)]\,\mathrm{d}\boldsymbol{Z}_i\mathrm{d}{\lambda}_i\\
   				&=\int J(\boldsymbol{Z}_i\mid\boldsymbol{Y}_i)E_{{\lambda}_i\sim q}[\log \mathbb{W}(\boldsymbol{Z}_i,{\lambda}_i,\boldsymbol{Y}_i)]\,\mathrm{d}\boldsymbol{Z}_i-\int J(\boldsymbol{Z}_i\mid\boldsymbol{Y}_i) \log J(\boldsymbol{Z}_i\mid\boldsymbol{Y}_i)\,\mathrm{d}\boldsymbol{Z}_i+constant\\
   				&=\int J(\boldsymbol{Z}_i\mid\boldsymbol{Y}_i)\log(\frac{\widetilde{J}(\boldsymbol{Z}_i\mid\boldsymbol{Y}_i)}{J(\boldsymbol{Z}_i\mid\boldsymbol{Y}_i)})\,\mathrm{d}\boldsymbol{Z}_i+constant\\
   				&=-D_{\mathrm{KL}}(J(\boldsymbol{Z}_i\mid\boldsymbol{Y}_i))\,\big\|\,\widetilde{J}(\boldsymbol{Z}_i\mid\boldsymbol{Y}_i))+constant,
   			\end{align*}
   		where
   		\begin{equation}
   			\log \widetilde{J}(\boldsymbol{Z}_i\mid\boldsymbol{Y}_i) = \mathbb{E}_{\boldsymbol{\lambda}\sim q} \left[ \log \mathbb{W}(\boldsymbol{Z}_i,{\lambda}_i,\boldsymbol{Y}_i) \right] + \mathrm{constant.},
   		\end{equation}
   		and the constant term is used to normalize $\widetilde{J}(\boldsymbol{Z}_i\mid\boldsymbol{Y}_i)$ into a valid probability density function and is independent of $\mathbb{Z}$. By the nonnegativity of the Kullback--Leibler divergence, we have
   		\begin{equation}
   			-D_{\mathrm{KL}} \left[{J}(\boldsymbol{Z}_i\mid\boldsymbol{Y}_i) \,\big\|\, \widetilde{J}(\boldsymbol{Z}_i\mid\boldsymbol{Y}_i) \right] \le 0.
   		\end{equation}
   		Therefore, the maximum value of this term is attained if and only if
   		\begin{equation}
   			{J}(\boldsymbol{Z}_i\mid\boldsymbol{Y}_i) = \widetilde{J}(\boldsymbol{Z}_i\mid\boldsymbol{Y}_i),
   		\end{equation}
   		under which the ELBO $\mathcal{L}(J,q)$ is maximized. Consequently, for a given $q({\lambda}_i\mid\boldsymbol{Y}_i)$, the optimal variational density $J^*(\boldsymbol{Z}_i\mid\boldsymbol{Y}_i)$ satisfies
   		\begin{align}\label{a14}
   				J^*(\boldsymbol{Z}_i\mid\boldsymbol{Y}_i) \propto \exp \left\{ \mathbb{E}_{{\lambda}_i\sim q} \left[ \log \mathbb{W}(\boldsymbol{Z}_i,{\lambda}_i,\boldsymbol{Y}_i) \right] \right\} .
   		\end{align}
   		Similarly, for a given $J(\boldsymbol{Z}_i\mid\boldsymbol{Y}_i)$, the optimal variational density $q^*({\lambda}_i\mid\boldsymbol{Y}_i)$ satisfies
   		\begin{align}\label{a15}
   			q^*({\lambda}_i\mid\boldsymbol{Y}_i) \propto \exp \left\{ \mathbb{E}_{\boldsymbol{Z}_i\sim J} \left[ \log \mathbb{W}(\boldsymbol{Z}_i,{\lambda}_i,\boldsymbol{Y}_i) \right] \right\} .
   		\end{align}
   	By substituting the expression of $p(\boldsymbol{Z}_i,{\lambda}_i,\boldsymbol{Y}_i)$ into \eqref{a14} and \eqref{a15},we have
   	\begin{align}
   		&J^*(\boldsymbol{Z}_i\mid\boldsymbol{Y}_i)\propto \prod_{j=1}^{m_i}\left[ J_0^*(\Delta Z_{i,j}\mid\boldsymbol{Y}_i)\prod_{k=1}^{K}J_k^*(\Delta y_{i,k,j}-\Delta Z_{i,j}\mid\boldsymbol{Y}_i)\right] ,\label{a16}\\
   		&q^*({\lambda_i}\mid\boldsymbol{Y}_i)\propto\exp\left\lbrace -\left[\sum_{j=1}^{m_i}\left( \sum_{k=1}^{K}E_{Z_i\sim J}[\Delta y_{i,k,j}-\Delta Z_{i,j}]+E_{Z_i\sim J}[\Delta Z_{i,j}]\right) +\rho \right]\lambda_i \right\rbrace \nonumber\\
   		&\quad\quad\quad\quad\quad\quad\times\lambda_i^{\sum_{j=1}^{m_i}\left( \sum_{k=1}^{K}\Delta\eta_{i,k,j} +\Delta\eta_{i,0,j}\right) +\omega-1}.\label{a17}
   	\end{align}
   	The derivation of \eqref{a16} and \eqref{a17} is presented in Appendix \ref{A1}. 
   	As we can see from \eqref{a16} and \eqref{a17}, $J_0^*(\Delta Z_{i,j}\mid\boldsymbol{Y}_i)$, $J_k^*(\Delta y_{i,k,j}-\Delta Z_{i,j}\mid\boldsymbol{Y}_i)$ and $q^*({\lambda_i}\mid\boldsymbol{Y}_i)$ are also gamma distributed, namely,
   	\begin{align*}
   		&J_0^*(\Delta Z_{i,j}\mid\boldsymbol{Y}_i)\sim Ga(\Delta\eta_{i,0,j},\mathbb{E}_{\lambda_i\sim q}[\lambda_i]),\\
   		&J_k^*(\Delta y_{i,k,j}-\Delta Z_{i,j}\mid\boldsymbol{Y}_i)\sim Ga(\Delta\eta_{i,k,j},\mathbb{E}_{\lambda_i\sim q}[\lambda_i]),\\
   		&q^*({\lambda_i}\mid\boldsymbol{Y}_i)\sim Ga\left( \sum_{j=1}^{m_i}\left( \sum_{k=1}^{K}\Delta\eta_{i,k,j} +\Delta\eta_{i,0,j}\right) +\omega,\sum_{j=1}^{m_i}\left( \sum_{k=1}^{K}\mathbb{E}_{Z_i\sim J}[\Delta y_{i,k,j}-\Delta Z_{i,j}]+\mathbb{E}_{Z_i\sim J}[\Delta Z_{i,j}]\right) +\rho\right).
   	\end{align*}
   	
   	However, according to \eqref{a16}, although 
   	$J_0^*(\cdot)$ and $J_k^*(\cdot)$ both have Gamma-kernel forms, 
   	the variational density of $\Delta Z_{i,j}$ induced by \eqref{a16} 
   	is generally not a standard Gamma density. Therefore, at the $r$th iteration of the EM algorithm, for a given $\bar{\lambda}_i=\mathbb{E}_{\lambda_i\sim q}[\lambda_i]$, define
   	\begin{align*}
   		J_{i,j}^{(r)}(z)=\frac{c_{i,j}^{(r)}(z)}{C_{i,j}^{(r)}}
   	\end{align*}
   	where 
   	\begin{align*}
   		&c_{i,j}^{(r)}(z)=z^{\Delta\eta_{i,0,j}^{(r)}-1}\prod_{k=1}^{K}
   		(\Delta y_{i,k,j}-z)^{\Delta\eta_{i,k,j}^{(r)}-1}\exp\left\{-\bar{\lambda}_i	\left[z+\sum_{k=1}^{K}(\Delta y_{i,k,j}-z)\right]\right\},\\
   		&C_{i,j}^{(r)}=\int_0^{\Delta\tilde y_{i,j}}c_{i,j}(z)\,\mathrm{d}z, \quad \text{and} ~0<z<\tilde y_{i,j}.
   	\end{align*}
   	
   	Therefore, at the $r$th iteration of the EM algorithm, given the current parameter estimate $\boldsymbol{\Theta}^{(r)}$, the optimal variational densities $J^*(\boldsymbol{Z}_i\mid\boldsymbol{Y}_i,\boldsymbol{\Theta}^{(r)})$ and $q^*(\lambda_i\mid\boldsymbol{Y}_i,\boldsymbol{\Theta}^{(r)})$ can be obtained through iterative updates. Algorithm \ref{alg:VI} summarizes the implementation procedure of this variational inference scheme. 
   		\begin{algorithm}
   			\caption{The VI procedure for the conditional distribution $\mathbb{W}(\boldsymbol{Z}_i,{\lambda}_i\mid\boldsymbol{Y}_i)$.}
   			\label{alg:VI}
   			\begin{algorithmic}
   				\Statex \textbf{Input:} Current estimates:$\boldsymbol{\Theta}^{(r)}= \{\alpha_0^{(r)},~\alpha_k^{(r)},~\beta_k^{(r)},~\rho^{(r)}~,\omega^{(r)}\}$; Observed degradation increments $\Delta y_{i,k,j}~i=1,2,\dots,n,~k=1,2,\dots,K,~j=1,2,\dots,m_i$; Error tolerance $\epsilon_2>0$.
   				\Statex \textbf{Output:} $J_0^*(\Delta Z_{i,j}\mid\boldsymbol{Y}_i,\boldsymbol{\Theta}^{(r)}),~J_k^*(\Delta y_{i,k,j}-\Delta Z_{i,j}\mid\boldsymbol{Y}_i,\boldsymbol{\Theta}^{(r)})$ and $q^*({\lambda_i}\mid\boldsymbol{Y}_i,\boldsymbol{\Theta}^{(r)})$.
   				\State Compute $A_i^{(r)}=\sum_{j=1}^{m_i}\left( \sum_{k=1}^{K}\Delta\eta^{(r)}_{i,k,j} +\Delta\eta^{(r)}_{i,0,j}\right) +\omega^{(r)}$;
   				\State Initialize $B_i^{\langle 0\mid r\rangle}={\rho^{(r)}}$, $\xi_{i}^{\langle0\mid r\rangle}=A_i^{(r)}/B_i^{\langle0\mid r\rangle}$, $d=1$ and $\text{error}=\epsilon_2+1$
   				
   				\While{$\text{error} > \epsilon_2$}
   				\For{$j=1,2,\dots,m_i$}
   				\State Compute 
   				\begin{align*}
   					c_{i,j}^{\langle d\mid r\rangle}(z)
   					&=
   					z^{\Delta\eta_{i,0,j}^{(r)}-1}
   					\prod_{k=1}^{K}
   					(\Delta y_{i,k,j}-z)^{\Delta\eta_{i,k,j}^{(r)}-1}  \\
   					&\quad \times
   					\exp\left\{
   					-\xi_i^{\langle d-1\mid r\rangle}
   					\left[
   					z+
   					\sum_{k=1}^{K}
   					(\Delta y_{i,k,j}-z)
   					\right]
   					\right\},
   					\qquad
   					0<z<\Delta\tilde y_{i,j};
   				\end{align*}
   				\State Compute $C_{i,j}^{\langle d\mid r\rangle}=\int_0^{\Delta\tilde y_{i,j}}c_{i,j}^{\langle d\mid r\rangle}(z)\,\mathrm{d}z;$
   				\State Compute $J_{i,j}^{\langle d\mid r \rangle}={c_{i,j}^{\langle d\mid r\rangle}(z)}/{C_{i,j}^{\langle d\mid r\rangle}}$;
   				\State Compute $\mu_{i,j}^{\langle d\mid r\rangle}=\int_0^{\Delta\tilde y_{i,j}}zJ_{i,j}^{\langle d\mid r\rangle}(z)\,\mathrm{d}z$
   				\For{$k=1,2,\dots,K$}
   				\State Compute $\tau_{i,k,j}^{\langle d\mid r\rangle}=\Delta y_{i,k,j}-\mu_{i,j}^{\langle d\mid r\rangle}$;
   				\EndFor
   				\EndFor
   				\State Compute $B_i^{\langle d\mid r\rangle}=\sum_{j=1}^{m_i}\left[ \mu_{i,j}^{\langle d\mid r\rangle}+\sum_{k=1}^{K}\tau_{i,k,j}^{\langle d\mid r\rangle}\right] $;
   				\State Compute $\xi_i^{\langle d\mid r\rangle}={A_i{(r)}}/B_i^{\langle d\mid r\rangle}$;
   				\State Compute $\text{error} = \max\left\{ \left| \xi_i^{<d \mid r>} - \xi_i^{<d-1 \mid r>} \right|, \left\{ \left| B_{i}^{<d \mid r>} - B_{i}^{<d-1 \mid r>} \right| \right\} \right\}$
   				\State Set $d \leftarrow d + 1$
   				\EndWhile
   				
   				\State Determine the optimal $J^*(\boldsymbol{Z}_i\mid\boldsymbol{Y}_i,\boldsymbol{\Theta}^{(r)})$ and $q^*(\lambda_i\mid\boldsymbol{Y}_i,\boldsymbol{\Theta}^{(r)})$  by:
   				\State \quad $J^*(\boldsymbol{Z}_i\mid\boldsymbol{Y}_i,\boldsymbol{\Theta}^{(r)})=\prod_{j=1}^{m_i}J_{i,j}^{\langle d-1\mid r\rangle}(\Delta z_{i,j})$,
   				\State \quad $q^*(\lambda_i\mid\boldsymbol{Y}_i,\boldsymbol{\Theta}^{(r)})\sim Ga(A_i^{(r)},B_i^{\langle d-1\mid r\rangle})$.
   			\end{algorithmic}
   		\end{algorithm}
   		
   		After running Algorithm \ref{alg:VI}, we can calculate the conditional expectation $\zeta_i^{(r)}$, $\xi_i^{(r)}$, $\pi_{i,j}^{(r)}$, $\sigma_{i,j}^{(r)}$ and $\nu_{i,j}^{(r)}$ by
   	\begin{align*}
   		&\zeta_i^{(r)}\approx\psi\left( A_i^{(r)}\right) -\log(B_{i}^{\langle d-1\mid r\rangle}),\\
   		&\xi_i^{(r)}\approx\frac{A_i^{(r)}}{B_{i}^{\langle d-1\mid r\rangle}},\\
   		&\pi_{i,j}^{(r)}\approx{\xi_i^{\langle d-1\mid r\rangle} \mu_{i,j}^{\langle d-1\mid r\rangle}},\\
   		&\sigma_{i,k,j}^{(r)}\approx\int_{0}^{\Delta\tilde y_{i,j}}\log(\Delta y_{i,k,j}-z)J_{i,j}^{\langle d-1\mid r\rangle}\,\mathrm{d}z,\\
   		&\nu_{i,j}^{(r)}\approx\int_{0}^{\Delta\tilde y_{i,j}}\log z J_{i,j}^{\langle d-1\mid r\rangle}\,\mathrm{d}z.\\
   	\end{align*}
   	\subsection{Interval estimation}\label{p5}
   	This section focuses on constructing interval estimates for the parameter vector 
   	$\boldsymbol{\Theta}$ and its function $\mathcal{X}(\boldsymbol{\Theta})$ using the 
   	parametric bootstrap method. The detailed bootstrap procedure is summarized in 
   	Algorithm~\ref{alg:bootstrap}. For the proposed model, the analytical sampling distribution of the 
   	parameter estimators is generally difficult to derive, because of the presence of 
   	latent variables, shared random effects, and posterior expectations that are not available in closed form. Moreover, the asymptotic normal approximation based on 
   	the Fisher information matrix may be inaccurate in finite samples or for complex 
   	hierarchical models. Therefore, the parametric bootstrap method is adopted in 
   	this study. By repeatedly generating pseudo-samples from the fitted model and 
   	re-estimating the parameters, this method provides a practical way to approximate 
   	the sampling uncertainty of the estimators(\cite{zhai2023multivariate}, \cite{zhai2025modeling} and \cite{zhuang2025multivariate}). It avoids the explicit derivation of 
   	asymptotic variances and can naturally account for the combined effects of the 
   	model structure, the estimation algorithm, and finite-sample variability.
   		
   	To obtain an approximate $100(1-\zeta)\%$ bootstrap confidence interval (CI) 
   	for the parameter function $\mathcal{X}(\boldsymbol{\Theta})$, the following steps are performed. First, $B$ bootstrap estimates are generated according to Algorithm~\ref{alg:bootstrap}:
   		\[
   		\left\{
   		\widehat{\boldsymbol{\Theta}}_1^*,
   		\widehat{\boldsymbol{\Theta}}_2^*,
   		\ldots,
   		\widehat{\boldsymbol{\Theta}}_B^*
   		\right\}.
   		\]
   	Then, the corresponding bootstrap values of the parameter function are computed as
   	\begin{align*}
   		\mathcal{X}_b^*=\mathcal{X}(\widehat{\boldsymbol{\Theta}}_b^*),\quad b=1,2,\dots,B.
   	\end{align*}
   	Let
   	\begin{align*}
   		\mathcal{X}_{(1)}^*\leq\mathcal{X}_{(2)}^*\leq\dots\leq\mathcal{X}_{(B)}^*
   	\end{align*}
   	denote the ordered bootstrap statistics. For a given significance level $\zeta$, the approximate $100(1-\zeta)\%$ bootstrap CI for $\mathcal{X}(\boldsymbol{\Theta})$ is given by 
   	\begin{align*}
   		\left[
   		\mathcal{X}_{(\lceil B\zeta/2\rceil)}^*,
   		\,
   		\mathcal{X}_{(\lceil B(1-\zeta/2)\rceil)}^*
   		\right],
   	\end{align*}
   	where $\mathcal{X}_{(p)}^*$ denotes the $p$th order statistic of the 
   	ordered bootstrap values, and $\lceil\cdot\rceil$ is the ceiling function used 
   	to determine the nearest upper integer position. This procedure provides an 
   	intuitive and implementable approach for quantifying the estimation uncertainty 
   	of $\mathcal{X}(\boldsymbol{\Theta})$.
   		\begin{algorithm}[H]
   			\caption{Parametric bootstrap procedure.}
   			\label{alg:bootstrap}
   			\begin{algorithmic}
   				\State \textbf{Input:} Point estimate $\hat{\boldsymbol{\Theta}}$, number of bootstrap replications $B$.
   				\State \textbf{Output:} Bootstrap estimates
   				\begin{align*}
   					\left\{\hat{\boldsymbol{\Theta}}_1^*, \hat{\boldsymbol{\Theta}}_2^*, \dots, \hat{\boldsymbol{\Theta}}_B^*\right\}.
   				\end{align*}
   				\For{$b = 1$ \textbf{to} $B$}
   				\For{$i = 1$ \textbf{to} $n$}
   				\State Generate the unit-specific shared random effect $\lambda_i^{*(b)}\sim Ga(\hat{\omega},\hat{\rho})$.
   				\For{$j = 1$ \textbf{to} $m_i$}
   				\State Generate the common shock-induced degradation increment 
   				\begin{align*}
   					\Delta {Z}_{i,j}^{*(b)}\sim Ga(\Delta \hat{\eta}_{i,0,j}, \lambda_i^{*(b)}).
   				\end{align*}
   				\For{$k = 1$ \textbf{to} $K$}
   				\State Generate the intrinsic degradation increment 
   				\begin{align*}
   					\Delta {X}_{i,k,j}^{*(b)}\sim Ga(\Delta \hat{\eta}_{i,k,j}, \lambda_i^{*(b)}).
   				\end{align*}
   				\State Construct the bootstrap total degradation increment
   				 \begin{align*}
   				 	\Delta {Y}_{i,k,j}^{*(b)} = \Delta {X}_{i,k,j}^{*(b)} + \Delta {Z}_{i,j}^{*(b)}
   				 \end{align*}
   				\EndFor
   				\EndFor
   				\EndFor
   				\State Construct the bootstrap total degradation increment dataset
   				\begin{align*}
   					\Delta \boldsymbol{Y}^{*(b)}=\{\Delta {Y}_{i,k,j}^{*(b)};i=1,2,\dots,n,k=1,2,\dots,K,j=1,2,\dots,m_i\}.
   				\end{align*}
   				\State Estimate the model parameters from $\Delta \boldsymbol{Y}^{*(b)}$ using the proposed EM algorithm, and obtain $\hat{\boldsymbol{\Theta}}_b^*$.
   				\EndFor
   				\State Return
   				\begin{align*}
   					\left\{\hat{\boldsymbol{\Theta}}_1^*, \hat{\boldsymbol{\Theta}}_2^*, \dots, \hat{\boldsymbol{\Theta}}_B^*\right\}.
   				\end{align*}
   			\end{algorithmic}
   		\end{algorithm}
\section{Maintenance optimization policy}\label{section5}
In the classical system replacement optimization problem, the system is replaced either upon failure or at a predetermined preventive replacement time $T$, whichever occurs first. Let $c_f$ and $c_p$ denote the costs of corrective replacement upon system failure and preventive replacement at time $T$, respectively. Since an unexpected system failure may incur additional loss costs, it is generally assumed that $c_f>c_p$. Let $T_g$ denote the system lifetime derived from the degradation model proposed in this manuscript. According to renewal reward theory, the long run expected cost rate per unit time can be expressed as:
\begin{align*}
	C(T)&=\frac{c_pP(T_g>T)+c_fP(T_g\leq T)}{\mathbb{E}[\min(T_g,T)]}\\
	&=\frac{c_f-(c_f-c_p)R(T)}{\int_{0}^{T}R(t)\mathrm{d} t}.
\end{align*}

For given $c_f$ and $c_p$, the maintenance optimization problem can be formulated as determining the optimal replacement time $T^*$ by minimizing the long run expected cost rate $C(T)$. In general, this one dimensional optimization problem can be solved using numerical optimization methods such as Brent's method. In particular, if the hazard rate function $h(t)$ of the system lifetime $T_g$ is increasing with $t$ to $\infty$, then a unique optimal replacement time \(T^*\) exists and satisfies
\begin{align*}
	h(T^*)\int_{0}^{T^*}R(t)\,dt+R(T^*) =\frac{c_f}{c_f-c_p}.
\end{align*}

Therefore, under these conditions, $T^*$ can be obtained by numerically solving the above optimality equation, and the corresponding minimum long run expected cost rate is
\begin{align*}
	C(T^*)=(c_f-c_p)h(T^*).
\end{align*}
\section{Simulation study}\label{section6}
    \subsection{Validation of model parameter estimators}
	To evaluate the effectiveness of the proposed method, a Monte Carlo simulation study is conducted based on degradation data generated from model [\ref{a21}]. We consider a multivariate degradation process with three performance characteristics (PCs), that is, $K = 3$. The shape functions $\eta_k(t)$ and $\eta_0(t)$ are specified as $\eta_k(t) = \beta_k t^{\alpha_k}, ~\eta_0(t) = t^{\alpha_0}$,
	respectively. The true parameter values are set as $\boldsymbol{\alpha} = (1.2, 1.8, 1.5)', ~ \boldsymbol{\beta} = (3.0, 3.5, 4.0)'$, $\alpha_0 = 0.8$, $\rho = 1$ and $\omega = 4$. Assume that the degradation processes of $n$ units are periodically measured with an interval of $t = 1$, and all units have the same number of measurements, that is, $m_1 = \dots = m_n \equiv m$.
	To investigate the effects of sample size and measurement frequency on parameter estimation, we consider configurations with $n=10, 30, 50$ and $m=20, 40$, namely, $(n,m)=(10,20), (10,40), (30,20), (30,40), (50,20)$, and $(50,40)$.
    For each $(n, m)$ configuration, 500 simulated degradation datasets are generated, and the proposed EM algorithm incorporating variational inference (VI) is applied to estimate the model parameters. A parametric bootstrap procedure with $B = 500$ bootstrap resamples is also implemented for each dataset to construct 95\% confidence intervals. Based on these results, the relative bias (RB) and relative root mean squared error (RRMSE) of the point estimates, together with the coverage probabilities of the 95\% confidence intervals, are calculated. For any true parameter $\theta \in \boldsymbol{\Theta}$, the RB and RRMSE are defined as
    \begin{equation*}
    	\operatorname{RB}(\hat{\theta}) = \frac{1}{500} \sum_{r=1}^{500} \frac{\hat{\theta}^{(r)} - \theta}{\theta},
    \end{equation*}
    and
    \begin{equation*}
    	\operatorname{RRMSE}(\hat{\theta}) = \left[ \frac{1}{500} \sum_{r=1}^{500} \left( \frac{\hat{\theta}^{(r)} - \theta}{\theta} \right)^2 \right]^{1/2},
    \end{equation*}
    where $\hat{\theta}^{(r)}$ denotes the estimate obtained from the $r$th Monte Carlo replication. Smaller absolute values of RB and smaller values of RRMSE indicate higher estimation accuracy and better estimation stability.
    
	Table \ref{table1} reports the RBs and RRMSEs of the point estimates, together with the CPs of the corresponding 95\% parametric bootstrap confidence intervals, under six combinations of $(n,m)$. The results show that the estimates of $\alpha_0,\alpha_1,\alpha_2,\alpha_3$ and $\beta_1,\beta_2,\beta_3$ exhibit relatively small biases across all settings, with the absolute RBs not exceeding $0.03$. In contrast, the parameters $\rho$ and $\omega$, which characterize heterogeneity across units, show relatively large positive biases when $n=10$, with RBs of approximately $0.40$–$0.48$. As the number of test units increases, the biases of both parameters decrease markedly and fall below $0.08$ when $n=50$. This indicates that accurate estimation of the heterogeneity parameters relies more strongly on having a sufficient number of independent test units. For a fixed $m$, the RRMSEs of all parameter estimates decrease as $n$ increases, with particularly pronounced improvements for $\rho$ and $\omega$. For a fixed $n$, increasing the number of observations within each unit, $m$, generally reduces the RRMSEs of $\alpha_0$, the $\alpha_k$'s, and the $\beta_k$'s. However, the improvement for $\rho$ and $\omega$ is relatively limited, and slight non-monotonic fluctuations are observed under some settings. These results suggest that, compared with increasing the observation frequency within each unit, increasing the number of independent test units provides more effective information for identifying and estimating heterogeneity across units.
	
    For interval estimation, the CPs reported in the table range from $0.887$ to $0.967$, with most values falling between $0.92$ and $0.97$, and are therefore generally close to the nominal confidence level of $0.95$. Overall, the EM algorithm combined with VI provides accurate and stable point estimates for the degradation path parameters. The estimation accuracy of the heterogeneity parameters also improves markedly when a moderate number of test units is available, while the parametric bootstrap procedure generally provides reasonable interval coverage performance.
	\begin{table*}[t]
		\centering
		\caption{RBs and RRMSEs of the point estimates, and CPs of the 95\% confidence intervals.}
		\label{table1}
		
		\footnotesize
		\renewcommand{\arraystretch}{0.95}
		
		\begin{tabular*}{\textwidth}{@{\extracolsep{\fill}}lccccccccc}
			\toprule
			\multirow{2}{*}{Parameter}
			& \multicolumn{3}{c}{(10,20)}
			& \multicolumn{3}{c}{(10,40)}
			& \multicolumn{3}{c}{(30,20)} \\
			\cmidrule(lr){2-4}
			\cmidrule(lr){5-7}
			\cmidrule(lr){8-10}
			& RB & RRMSE & CP
			& RB & RRMSE & CP
			& RB & RRMSE & CP \\
			\midrule
			$\alpha_0$ & 0.0265 & 0.1297 & 0.8867 & 0.0262 & 0.0753 & 0.9300 & 0.0192 & 0.0640 & 0.9400 \\
			$\alpha_1$ & -0.0012 & 0.0311 & 0.9667 & 0.0007 & 0.0204 & 0.9500 & -0.0006 & 0.0172 & 0.9600 \\
			$\alpha_2$ & -0.0004 & 0.0119 & 0.9533 & 0.0001 & 0.0065 & 0.9300 & 0.0007 & 0.0065 & 0.9367 \\
			$\alpha_3$ & 0.0004 & 0.0188 & 0.9233 & 0.0007 & 0.0101 & 0.9467 & 0.0006 & 0.0103 & 0.9567 \\
			$\beta_{1}$ & 0.0034 & 0.1329 & 0.9400 & -0.0199 & 0.1031 & 0.9200 & -0.0121 & 0.0747 & 0.9333 \\
			$\beta_{2}$ & 0.0056 & 0.0853 & 0.9567 & -0.0151 & 0.0607 & 0.9267 & -0.0140 & 0.0505 & 0.9467 \\
			$\beta_{3}$    & 0.0035 & 0.1034 & 0.9300 & -0.0179 & 0.0690 & 0.9567 & -0.0125 & 0.0595 & 0.9500 \\
			$\rho$    & 0.4748 & 1.0380 & 0.9167 & 0.4814 & 1.0501 & 0.9033 & 0.1205 & 0.3412 & 0.8967 \\
			$\omega$    & 0.4218 & 0.9822 & 0.9400 & 0.4006 & 0.8923 & 0.9067 & 0.0993 & 0.3098 & 0.9000 \\
			\bottomrule
		\end{tabular*}
		
		\vspace{4pt}
		
		\begin{tabular*}{\textwidth}{@{\extracolsep{\fill}}lccccccccc}
			\toprule
			\multirow{2}{*}{Parameter}
			& \multicolumn{3}{c}{(30,40)}
			& \multicolumn{3}{c}{(50,20)}
			& \multicolumn{3}{c}{(50,40)} \\
			\cmidrule(lr){2-4}
			\cmidrule(lr){5-7}
			\cmidrule(lr){8-10}
			& RB & RRMSE & CP
			& RB & RRMSE & CP
			& RB & RRMSE & CP \\
			\midrule
			$\alpha_0$ & 0.0255 & 0.0460 & 0.9500 & 0.0149 & 0.0492 & 0.9400 & 0.0256 & 0.0394 & 0.9600 \\
			$\alpha_1$ & 0.0005 & 0.0123 & 0.9500 & 0.0001 & 0.0143 & 0.9533 & 0.0009 & 0.0099 & 0.9333 \\
			$\alpha_2$ & 0.0003 & 0.0036 & 0.9400 & 0.0005 & 0.0053 & 0.9267 & 0.0002 & 0.0027 & 0.9567 \\
			$\alpha_3$ & 0.0002 & 0.0056 & 0.9667 & 0.0003 & 0.0077 & 0.9500 & 0.0007 & 0.0046 & 0.9400 \\
			$\beta_{1}$ & -0.0275 & 0.0641 & 0.9467 & -0.0139 & 0.0619 & 0.9367 & -0.0295 & 0.0576 & 0.9233 \\
			$\beta_{2}$ & -0.0230 & 0.0404 & 0.9300 & -0.0143 & 0.0411 & 0.9367 & -0.0227 & 0.0339 & 0.9300 \\
			$\beta_{3}$    & -0.0235 & 0.0448 & 0.9433 & -0.0135 & 0.0455 & 0.9367 & -0.0255 & 0.0402 & 0.9467 \\
			$\rho$    & 0.1821 & 0.3832 & 0.9233 & 0.0785 & 0.2603 & 0.8967 & 0.0763 & 0.2429 & 0.9467 \\
			$\omega$    & 0.1370 & 0.3222 & 0.9333 & 0.0619 & 0.2378 & 0.9100 & 0.0483 & 0.2223 & 0.9300 \\
			\bottomrule
		\end{tabular*}
		
	\end{table*}
	\subsection{Performance of reliability estimation}
   To evaluate the performance of the proposed method in reliability estimation, we calculate the mean time to failure (MTTF) based on the estimated model parameters. The failure thresholds for the three PCs are set to 30, that is, $\tau_1=\tau_2=\tau_3=30$. The MTTF of a unit is defined as
   \begin{equation*}
   	\mathrm{MTTF}
   	=\mathbb{E}(T)
   	=\int_{0}^{\infty}R(t)\,\mathrm{d}t,
   \end{equation*}
   where $R(t)$ denotes the reliability function of the unit. Under the specified true model parameters and failure thresholds, the true MTTF is calculated numerically as 6.7588.
		
	For each simulated dataset, the MTTF is calculated using the corresponding estimated model parameters and compared with the true MTTF obtained under the true parameter values. Figure \ref{fig2} presents boxplots of the absolute relative errors of the MTTF estimates under different simulation scenarios and sample size combinations. Consistent with the preceding parameter estimation results, the proposed method yields relatively small errors in the MTTF estimates. As the sample size increases, the MTTF estimation errors generally decrease, while their dispersion and the frequency of outliers also decline, indicating that both the accuracy and stability of reliability estimation improve as more sample information becomes available.
	\begin{figure}[htbp]
		\centering
		\includegraphics[
		width=\linewidth,
		height=0.85\textheight,
		keepaspectratio
		]{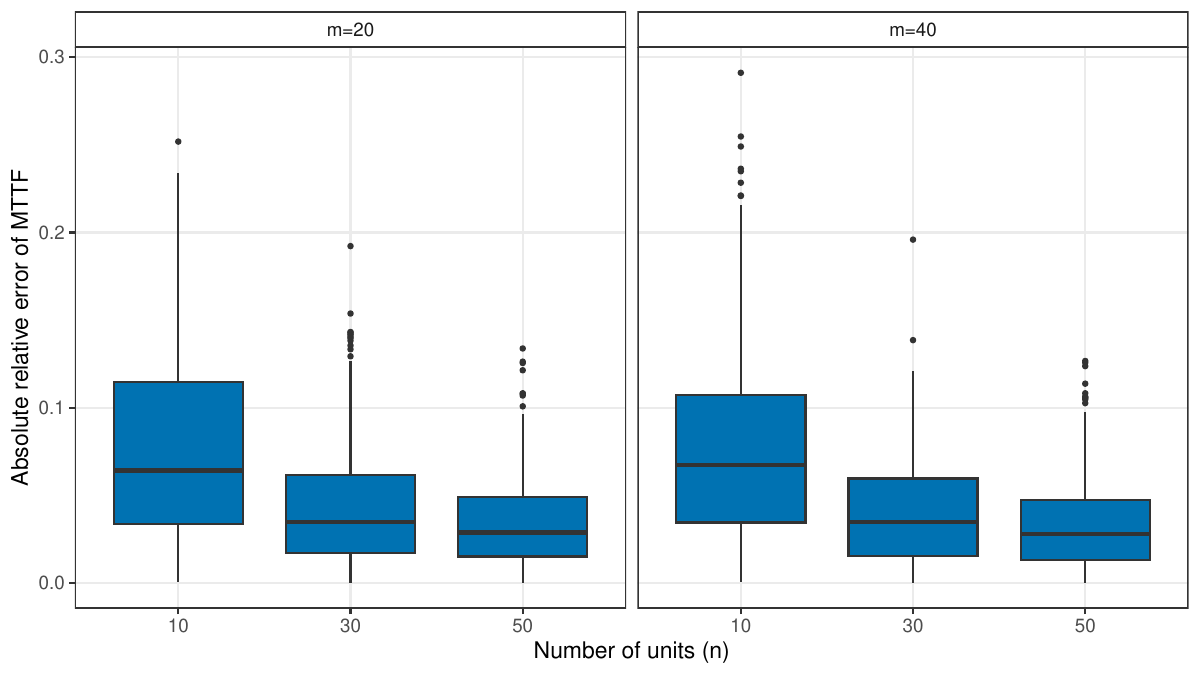}
		\caption{Absolute relative errors of the MTTF under different combinations of $(n,m)$.}
		\label{fig2}
	\end{figure}
	
    To further investigate the effect of model misspecification on reliability estimation, we consider two simplified scenarios. In Scenario 1, the unit specific random effect is omitted, while the dependence among PCs induced by the common environmental degradation process is retained:
	\begin{equation*}
			X_k(t) \sim \text{Ga}(\eta_k(t), \lambda), \qquad Z(t) \sim \text{Ga}(\eta_0(t), \lambda), \qquad k = 1, 2, \dots, K,
		\end{equation*}
	where $\lambda$ is treated as a fixed parameter shared by all test units. In Scenario 2, the degradation processes of the PCs are assumed to be mutually independent; that is, both the unit specific random effect and the common environmental degradation process are omitted. For both simplified models, the model parameters are estimated using the MLE. To ensure a consistent comparison among the three models, the same simulated degradation datasets are used under each sample size combination for reliability estimation based on the proposed model, Scenario 1, and Scenario 2, respectively. Figure~\ref{fig3} presents boxplots of the absolute relative errors of the MTTF estimates for the three models under different sample size combinations.
	\begin{figure}[htbp]
		\centering
		\includegraphics[
		width=\linewidth,
		height=0.85\textheight,
		keepaspectratio
		]{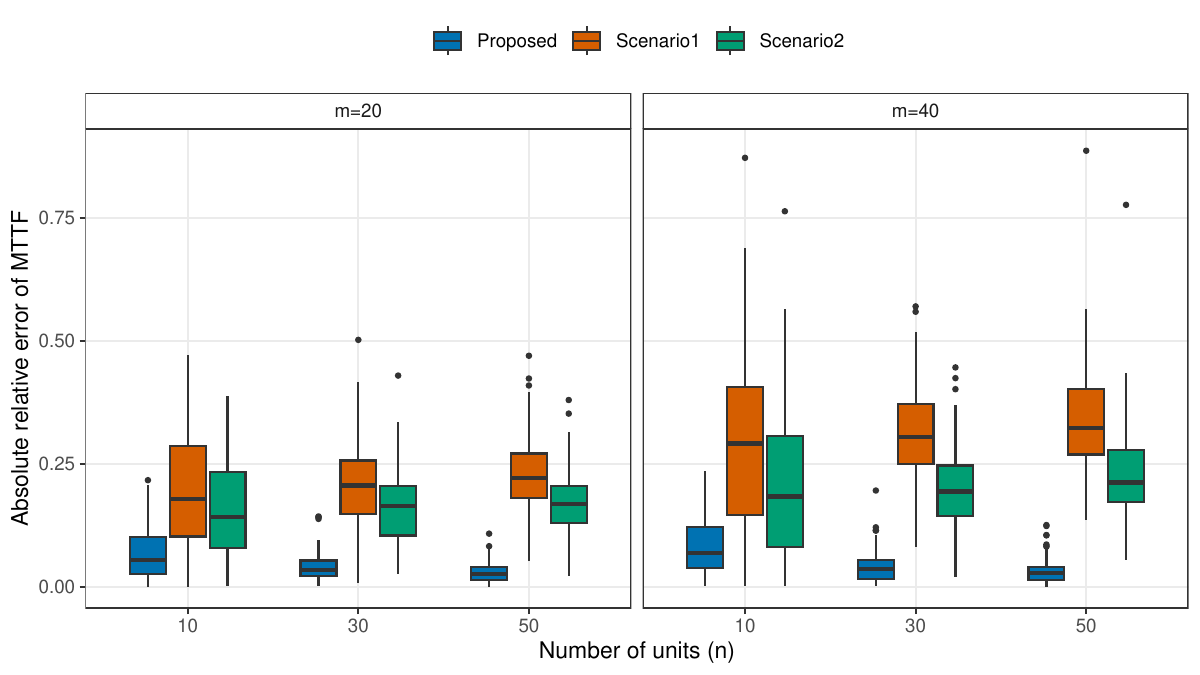}
		\caption{Absolute relative errors of MTTF under different scenarios.}
		\label{fig3}
	\end{figure}
		
	The numerical results show that the average RRMSEs of the MTTF estimates ($\times 10^{-2}$) under Scenarios 1 and 2 are 28.8832 and 21.3087, respectively, whereas the corresponding average RRMSE for the proposed model is only 6.1061. These results indicate that, for degradation data exhibiting both unit heterogeneity and dependence among performance characteristics, ignoring the unit specific random effect and the dependence structure among the PCs can substantially reduce the accuracy of reliability estimation. Inaccurate reliability estimates may further lead to inappropriate maintenance decisions, additional maintenance costs, and an increased risk of production interruptions.
\section{Application to the Coating Material Degradation Data}\label{section7}
In this case study, we analyze multivariate degradation data from an amine cured epoxy coating system. The data originate from experiments conducted by the National Institute of Standards and Technology(NIST) using the SPHERE accelerated weathering facility. Detailed descriptions of the experimental design and testing procedures can be found in Gu et al. (2009), while the machine readable data used in this analysis were obtained from the official supplementary materials made publicly available by Lu et al. (2021). During the experiments, coating specimens were subjected to accelerated weathering under strictly controlled ultraviolet(UV) radiation, temperature, and relative humidity conditions. Their chemical degradation was monitored at different exposure times using transmission Fourier transform infrared(FTIR) spectroscopy. Changes in the intensities of characteristic FTIR peaks reflect changes in the corresponding chemical structures during degradation. In this study, the following three peaks are selected as performance characteristics (PCs): 
\begin{itemize}
	\item the peak at $1250\,\mathrm{cm}^{-1}$ from C-O stretching of aryl ether, denoted by $Y_1$;
	
	\item the peak at $1510\,\mathrm{cm}^{-1}$ from benzene ring stretching, denoted by $Y_2$;
	
	\item the peak at $2925\,\mathrm{cm}^{-1}$ from CH$_2$ antisymmetric stretching, denoted by $Y_3$.
\end{itemize} 
Thus, three PCs are considered in this case, that is, \(K=3\).

The complete publicly available dataset contains 48 coating specimens under 12 combinations of experimental conditions. In this study, we select the condition with a band pass filter having a nominal center wavelength of 353 nm, a 100\% neutral density filter, a temperature of 35 ℃, and a relative humidity of 25\%. Four replicate units are available under this condition, namely NSS-0015-03, NSS-0015-07, NSS-0015-08, and NSS-0015-10, giving $n=4$. For each unit, all three PCs were measured at the same 25 irregularly spaced observation times, including the initial time $t=0$.

Figure \ref{fig1} shows the degradation measurements of the four experimental units starting from $t=0$. Overall, all three PCs exhibit increasing degradation trends with exposure time, while clear differences are observed in their degradation rates and cumulative degradation levels. On the other hand, the degradation paths of the three PCs within the same unit exhibit similar increasing trends and show relatively synchronous changes over some observation intervals, suggesting potential positive dependence among the PCs as well as correlated degradation induced by common influencing factors. In addition, the degradation paths differ to some extent in their cumulative degradation levels across units, and the magnitude of these differences varies among the PCs, reflecting unit heterogeneity. These observations motivate the use of a multivariate degradation model that can simultaneously characterize dependence among PCs and heterogeneity across units.
\begin{figure}[H]
	\centering
	\includegraphics[
	width=\textwidth,
	height=0.72\textheight,
	keepaspectratio
	]{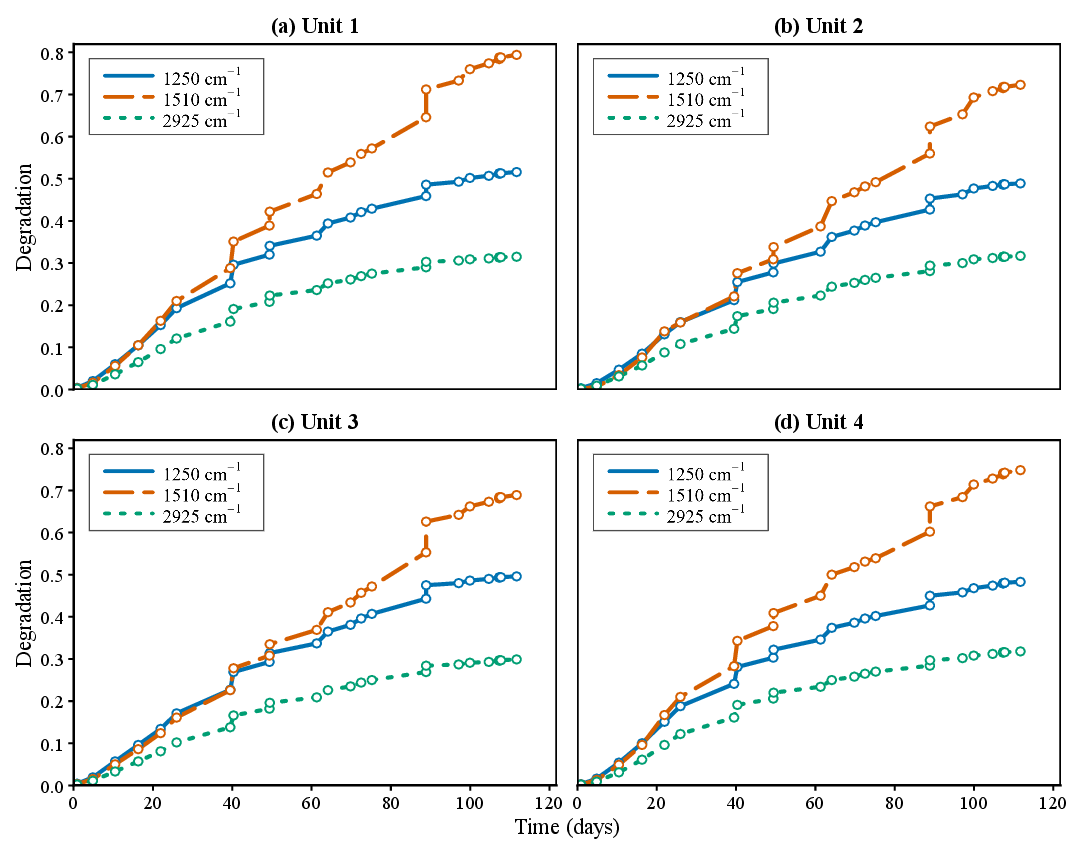}
	\vspace{-3mm}
	\caption{Degradation of the four units under the 353 nm, 100\% neutral density, 35~$^\circ$C, and 25\% relative humidity condition.}
	\label{fig1}
	\vspace{-4mm}
\end{figure}
To investigate the joint degradation behavior and dependence among the three PCs, the proposed multivariate Gamma degradation model is applied to the chemical degradation data measured at the three characteristic peaks of $1250$, $1510$, and $2925\,\mathrm{cm}^{-1}$. Since the observed degradation paths exhibit pronounced nonlinear trends, the power law functions $\eta_k(t)=\beta_k t^{\alpha_k},~ k=1,2,3,$ and $\eta_0(t)=t^{\alpha_0}$ are adopted to describe the intrinsic degradation shape functions of the PCs and the shape function of the shared environmental damage process, respectively. Table 2 reports the estimates of the model parameters $ \boldsymbol{\theta} =(\alpha_1,\alpha_2,\alpha_3,\beta_1,\beta_2,\beta_3,\alpha_0,\rho,\omega)',$ and the corresponding 90\% confidence intervals obtained using the parametric bootstrap method. Based on the parameter estimates in Table 2, the three PCs exhibit noticeably different mean degradation behaviors. Over the observed time range, $Y_2$ has the largest estimated mean degradation rate, followed by $Y_1$, whereas $Y_3$ has the smallest estimated mean degradation rate. This result is consistent with the degradation paths observed in Figure \ref{fig1}. In addition, according to Proposition \ref{pro1}, the estimated mean degradation levels of $Y_1$, $Y_2$, and $Y_3$ at the final observation time, $t=111.7$ days, are 0.5382, 0.7072, and 0.3205, respectively. These values exhibit the same ordering, $Y_2>Y_1>Y_3$, which is consistent with the above comparison of their mean degradation behaviors and further supports the ability of the proposed model to characterize the different degradation patterns of the three PCs.
\begin{table}[htbp]
	\centering
	\caption{Parameter estimates and 90\% parametric-bootstrap confidence intervals for the proposed multivariate Gamma degradation model.}
	\label{tab2}
	\small
	\renewcommand{\arraystretch}{1.08}
	\begin{tabular*}{\linewidth}{@{\extracolsep{\fill}}lcc@{}}
		\toprule
		& \textbf{Point estimates}
		& \textbf{90\% confidence interval} \\
		\midrule
		$\hat{\alpha}_0\;(\times 10^{-1})$
		& $7.5795$ & $[7.3259,\,7.8202]$ \\
		
		$\hat{\alpha}_1\;(\times 10^{-1})$
		& $7.5509$ & $[6.5569,\,8.5562]$ \\
		
		$\hat{\alpha}_2\;(\times 10^{0})$
		& $1.1504$ & $[1.0247,\,1.2792]$ \\
		
		$\hat{\alpha}_3\;(\times 10^{-1})$
		& $2.4904$ & $[1.1703,\,9.6102]$ \\
		
		\addlinespace[2pt]
		
		$\hat{\beta}_1\;(\times 10^{-1})$
		& $7.0902$ & $[4.2486,\,11.7351]$ \\
		
		$\hat{\beta}_2\;(\times 10^{-1})$
		& $1.9371$ & $[1.0652,\,3.5362]$ \\
		
		$\hat{\beta}_3\;(\times 10^{-1})$
		& $1.3339$ & $[0.0187,\,4.0720]$ \\
		
		\addlinespace[2pt]
		
		$\hat{\rho}\;(\times 10^{2})$
		& $3.6828$ & $[3.5987,\,4.5995]$ \\
		
		$\hat{\omega}\;(\times 10^{4})$
		& $4.1484$ & $[4.0579,\,5.1951]$ \\
		\bottomrule
	\end{tabular*}
\end{table}

   Figure \ref{fig4} presents the estimated mean degradation paths and the corresponding 90\% marginal point-wise prediction intervals for the degradation processes at the three characteristic peaks. As shown in the figure, the estimated mean degradation paths capture the nonlinear degradation trends of the three performance characteristics well, while the prediction intervals generally widen over time. This indicates that the uncertainty in future degradation levels increases as degradation damage accumulates, which is consistent with the stochastic evolution of a cumulative degradation process. Overall, the observed degradation paths of all three PCs are well covered by the corresponding 90\% prediction intervals, suggesting that the proposed model provides a reasonable description of both the observed degradation behavior and its associated uncertainty.
   \begin{figure}[htbp]
   	\centering
   	\includegraphics[
   	width=\linewidth,
   	height=0.85\textheight,
   	keepaspectratio
   	]{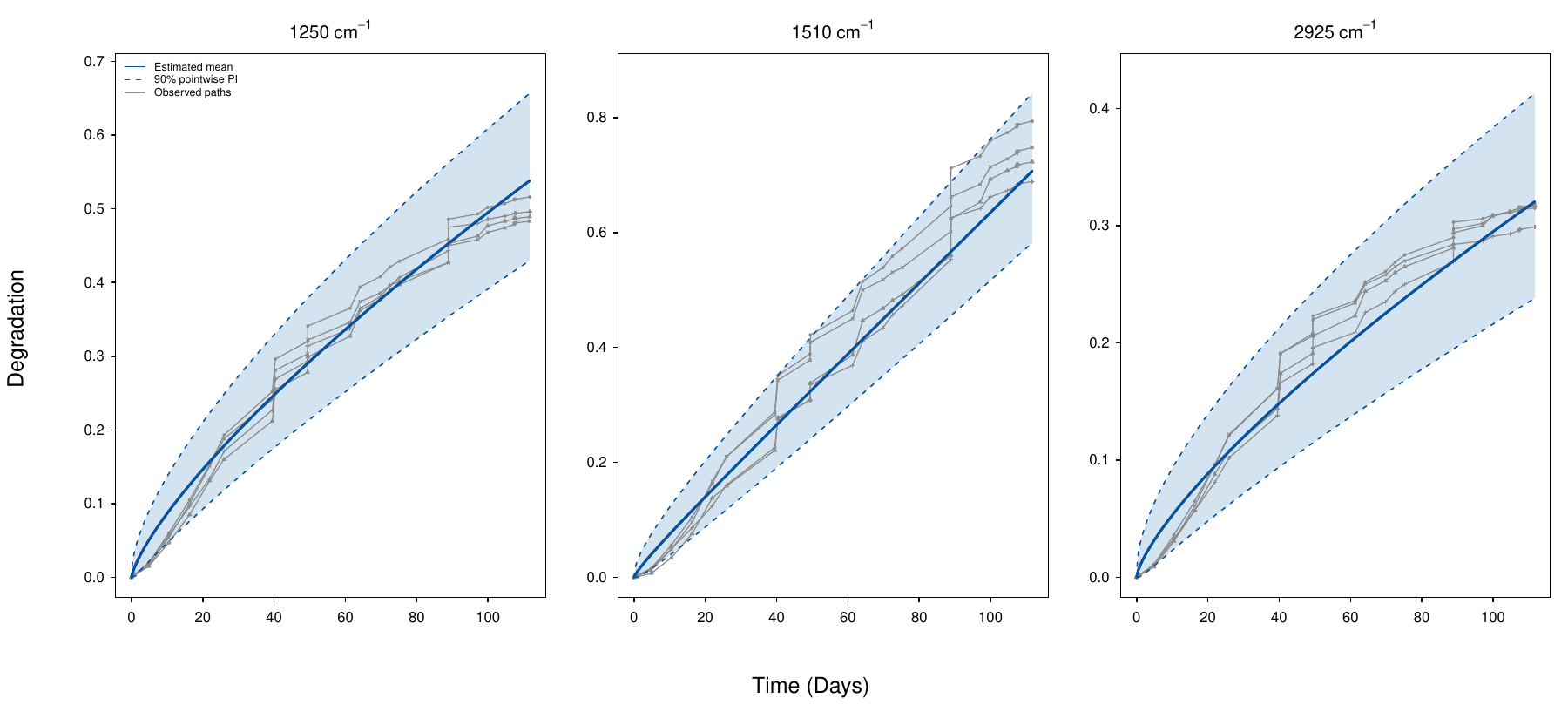}
   	\caption{The estimated mean degradation path and 90\% prediction band (PB) for the degradation at 1250 cm$^{-1}$, 1510 cm$^{-1}$, and 2925 cm$^{-1}$.}
   	\label{fig4}
   \end{figure}
   
   To further investigate the dependence among the three degradation processes, we examine the two sources of dependence incorporated in the proposed model, namely the shared environmental degradation process \(Z(t)\) and the shared random effect \(\lambda\). Conditional on \(\lambda\), the dependence among different performance characteristics is primarily induced by the common process \(Z(t)\). After integrating over \(\lambda\), the unconditional correlations additionally incorporate the dependence arising from latent heterogeneity across units. According to Proposition \ref{pro1} and the estimated model parameters, the conditional correlation coefficients among the three performance characteristics, given \(\lambda\), are estimated as
   
   \begin{align*} 
   	\hat{\rho}\{Y_1,Y_2\mid\lambda\} &=0.5133\;([0.4758,\,0.5470]),\\ \hat{\rho}\{Y_1,Y_3\mid\lambda\} &=0.7625\;([0.7314,\,0.7894]),\\ \hat{\rho}\{Y_2,Y_3\mid\lambda\} &=0.6652\;([0.6352,\,0.6929]). 
   	\end{align*}
   
   The corresponding unconditional correlation coefficients are estimated as
   \begin{align*} 
   	\hat{\rho}\{Y_1,Y_2\} &=0.5141\;([0.4764,\,0.5477]),\\ 
   	\hat{\rho}\{Y_1,Y_3\} &=0.7627\;([0.7316,\,0.7896]),\\ 
   	\hat{\rho}\{Y_2,Y_3\} &=0.6655\;([0.6355,\,0.6933]). 
   \end{align*}
   
   These results indicate clear positive dependence among the three performance characteristics. The strongest correlation is observed between the degradation processes at \(1250\) and \(2925\,\mathrm{cm}^{-1}\), followed by that between \(1510\) and \(2925\,\mathrm{cm}^{-1}\), while the correlation between \(1250\) and \(1510\,\mathrm{cm}^{-1}\) is relatively weaker. Notably, the three unconditional correlation coefficients are only slightly larger than their corresponding conditional counterparts, and the differences are very small. This suggests that, under the present data and estimated parameter values, the additional dependence induced by the shared random effect \(\lambda\) is relatively limited, whereas the dependence structure among the three performance characteristics is driven primarily by the shared environmental degradation process \(Z(t)\). In other words, although the shared random effect is still used to characterize latent heterogeneity across test units and thereby contributes additional dependence among the PCs, its contribution in this case is relatively weak compared with that induced by the shared environmental degradation process. These findings further demonstrate the importance of accounting for both the shared environmental degradation effect and the additional dependence arising from unit heterogeneity when modeling the joint degradation behavior at the three characteristic peaks.
   
   Following the related studies of \cite{lu2021general} and \cite{zhai2023multivariate}, the soft failure thresholds for the three performance characteristics (PCs) at $1250$, $1510$, and $2925~\mathrm{cm}^{-1}$ are set to $0.40$, $0.58$, and $0.27$, respectively; that is, $ (\tau_1,\tau_2,\tau_3)=(0.40,0.58,0.27)$. Under the failure criterion adopted in this manuscript, the system is regarded as failed once any PC reaches its corresponding failure threshold. Accordingly, the system lifetime is defined as $ T_g=\lim\limits_{k=1,2,3}\inf\{t:Y_k(t)\geq\tau_k\}$. Based on the parameter estimates reported in Table 4 and the reliability function in \ref{a2}, the mean time to failure (MTTF) of the system is estimated to be 75.3126 days. Figure 5 presents the corresponding system reliability function. As shown in the figure, the system reliability remains close to 1 during the early exposure period and then gradually decreases as degradation damage accumulates, reaching 0.5 at approximately 75.20 days. This behavior reflects the cumulative effect of the joint degradation of the three performance characteristics on the system failure risk and indicates that the likelihood of system failure increases progressively with exposure time.
   
   On this basis, we consider the corresponding preventive replacement policy. To examine the economic performance of this maintenance policy, the preventive replacement cost is normalized to $c_p=1$, while the failure replacement cost is set to $c_f=5$. The numerical results show that the optimal preventive replacement time is 49.5875 days. At this point, the system reliability is approximately 0.9685, and the corresponding minimum long run expected cost rate is 0.0228. Figure 5 also presents the long run expected cost rate as a function of the preventive replacement time. The curve exhibits a pronounced U shape. If replacement is scheduled too early, excessively frequent planned replacements increase the preventive maintenance cost. In contrast, if replacement is delayed too long, the probability of system failure before the planned replacement time increases, resulting in higher failure replacement costs. Therefore, the balance between preventive and failure replacement costs leads to a minimum long run expected cost rate at approximately 49.59 days.  Furthermore, if no preventive maintenance is implemented and the system is replaced only after failure, the long run cost rate is approximately 0.0664. Compared with this corrective replacement policy, the optimal preventive replacement policy reduces the long run cost rate by approximately $65.7\%$. These results indicate that a preventive replacement policy based on the reliability information obtained from the proposed multivariate degradation model can substantially reduce maintenance costs while controlling the risk of system failure, thereby providing a quantitative basis for maintenance decision making for the coating system.
   \begin{figure}[H]
   	\centering
   	\includegraphics[
   	width=\textwidth,
   	height=0.72\textheight,
   	keepaspectratio
   	]{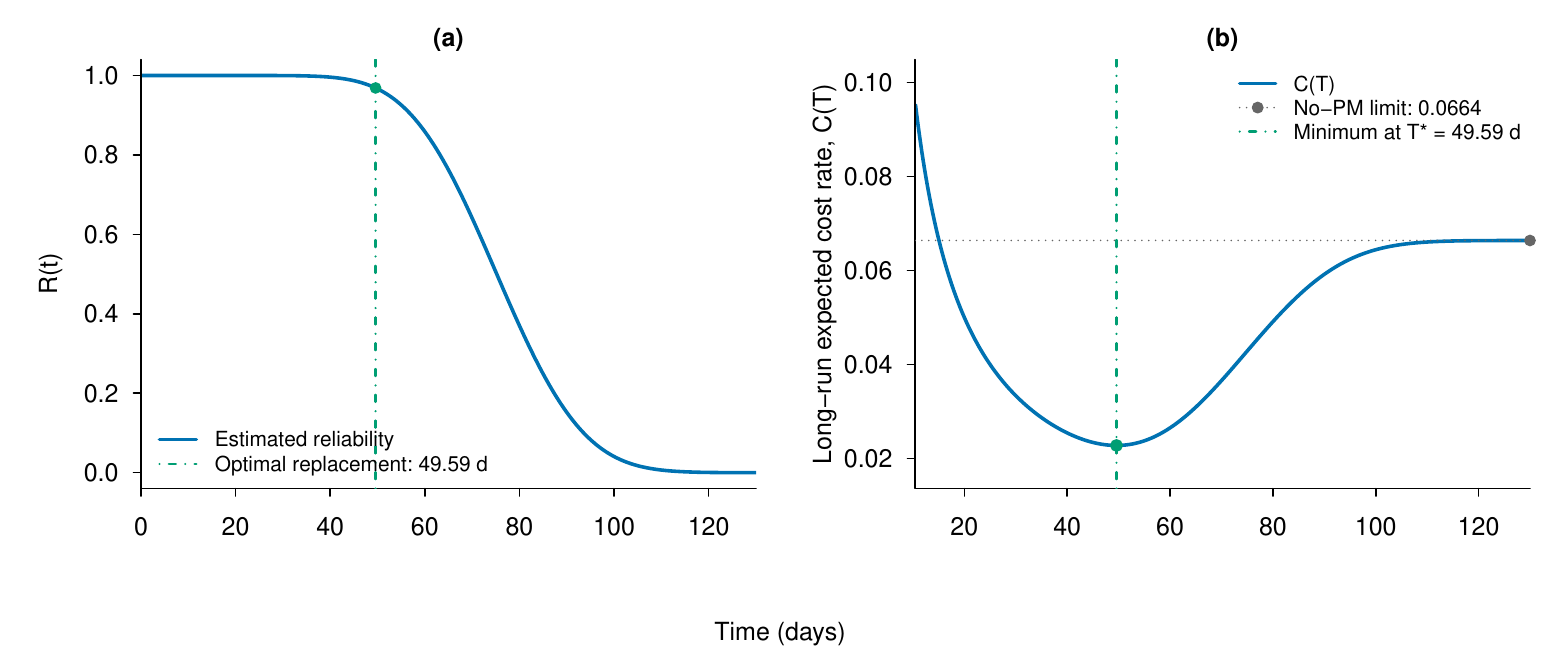}
   	\vspace{-3mm}
   	\caption{Estimated system reliability and long-run expected cost rate.}
   	\label{fig5}
   	\vspace{-4mm}
   \end{figure}
   
\section{Concluding remarks}\label{section8}
This study develops a multivariate Gamma process model for the joint analysis of multiple monotonically degrading performance characteristics subject to shared environmental damage and unit heterogeneity. While preserving the nonnegative increments of the degradation processes, the proposed model distinguishes two interpretable sources of dependence, namely the shared degradation component and the shared random effect. We further derive the marginal moments, covariance, and time varying correlation structure of the model, together with an explicit integral representation of the threshold based system reliability. A deterministic quadrature method is developed to enable stable and reproducible reliability evaluation. For statistical inference, the latent shared degradation amount and the random effect are treated as missing variables within an EM framework, while variational inference and deterministic quadrature are combined to evaluate the conditional expectations that are analytically intractable. In addition, a parametric bootstrap procedure is employed to quantify the uncertainty in parameter estimation.

Under the current model specification, the shared degradation process is assumed to affect all performance characteristics with the same loading. Therefore, in future research, when different performance characteristics exhibit different sensitivities to shared environmental damage, PC specific loading parameters or multiple shared Gamma factors may be introduced to enhance model flexibility. In addition, measurement error, operating covariates, heterogeneous initial degradation levels, and random degradation onset times may be incorporated to broaden the range of experimental settings and practical applications that can be accommodated by the model. Finally, the age replacement policy considered in this study could be further extended to condition based or dynamically updated maintenance policies that account for inspection costs, imperfect maintenance, downtime losses, and uncertainty in both failure thresholds and economic parameters.
\section*{Funding}
This research was supported by the National Natural Science Foundation of China (No. 12361060), the Science and Technology Program of Gansu Province (No. 24JRRA144), the College Teachers Innovation Foundation Project of the Gansu Provincial Department of Education (No. 2024A002), the Industrial Support Program of the Gansu Provincial Department of Education (No. 2025CYZC-016), and the Innovative Fundamental Research Group Project of Gansu Province (No. 25JRRA002).
\section*{Disclosure statement}
No potential conflict of interest was reported by the author(s).
\bibliographystyle{mystyle}

\setlength{\bibsep}{0.4ex}
\begin{spacing}{0.9}
\bibliography{ref-quan}
\end{spacing}
\appendix

\renewcommand{\thesection}{Appendix~A\arabic{section}}
\renewcommand{\thesubsection}{A\arabic{section}.\arabic{subsection}}
\renewcommand{\theequation}{A\arabic{equation}}
\newpage
\section{Derivation of \eqref{a16} and \eqref{a17}}
\label{A1}
\setcounter{equation}{0}
Note that $\log(\mathbb{W}(\boldsymbol{Z}_i,{\lambda}_i,\boldsymbol{Y}_i))=\log(\mathcal{L}(\boldsymbol{Y}_i\mid \boldsymbol{Z}_i,\lambda_i))+\log(\Psi(\boldsymbol{Z}_i\mid\lambda_i))+\log(\varphi(\lambda_i))$, where $\mathcal{L}(\boldsymbol{Y}_i\mid \boldsymbol{Z}_i,\lambda_i)$, $\Psi(\boldsymbol{Z}_i\mid\lambda_i)$ and $\varphi(\lambda_i)$ are given in \eqref{a18}, \eqref{a19} and \eqref{a20}, respectively. Then, we can derive $J^*(\boldsymbol{Z}_i\mid\boldsymbol{Y}_i)$ as follows
\begin{align*}
	J^*(\boldsymbol{Z}_i\mid\boldsymbol{Y}_i) &\propto \exp \left\{ \mathbb{E}_{{\lambda}_i\sim q} \left[ \log \mathbb{W}(\boldsymbol{Z}_i,{\lambda}_i,\boldsymbol{Y}_i) \right] \right\}\\
	&\propto\exp \left\lbrace -\sum_{j=1}^{m_i}\sum_{k=1}^{K}(\Delta y_{i,k,j}-\Delta z_{i,j})\mathbb{E}_{\lambda_i\sim q}[\lambda_i]+\sum_{j=1}^{m_i}\sum_{k=1}^{K}(\Delta \eta_{i,k,j}-1)\log(\Delta y_{i,k,j}-\Delta z_{i,j})\right.\\
	&\left.\qquad\qquad -\sum_{j=1}^{m_i}\Delta z_{i,j}\mathbb{E}_{\lambda_i\sim q}[\lambda_i]+\sum_{j=1}^{m_i}(\Delta \eta_{i,0,j}-1)\log(\Delta z_{i,j})\right\rbrace\\
	&=\prod_{j=1}^{m_i} \bigg\{ \Delta z_{i,j}^{\Delta \eta_{i,0,j}-1}\exp\{-\mathbb{E}_{\lambda_i\sim q}[\lambda_i]\Delta z_{i,j}\}\prod_{k=1}^{K}(\Delta y_{i,k,j}-\Delta z_{i,j})^{\Delta \eta_{i,k,j}-1}\\
	&\quad\times\exp\{-\mathbb{E}_{\lambda_i\sim q}[\lambda_i](\Delta y_{i,k,j}-\Delta z_{i,j})\} \bigg\} \\
	&\triangleq\prod_{j=1}^{m_i}\bigg\{J_0^*(\Delta Z_{i,j}\mid\boldsymbol{Y}_i)\prod_{k=1}^{K}J_k^*(\Delta y_{i,k,j}-\Delta Z_{i,j}\mid\boldsymbol{Y}_i)\bigg\}.
\end{align*}
On the other hand, the derivation of $q^*(\lambda_i\mid\boldsymbol{Y}_i)$ is given by
\begin{align*}
	q^*({\lambda}_i\mid\boldsymbol{Y}_i) &\propto \exp \left\{ \mathbb{E}_{\boldsymbol{Z}_i\sim J} \left[ \log \mathbb{W}(\boldsymbol{Z}_i,{\lambda}_i,\boldsymbol{Y}_i) \right] \right\}\\
	&\propto\exp\left\lbrace \left[\sum_{j=1}^{m_i}\left( \sum_{k=1}^{K}\Delta\eta_{i,k,j} +\Delta\eta_{i,0,j}\right) +\omega-1\right]\log\lambda_i\right.\\
	&\left.\qquad\quad-\left[ \sum_{j=1}^{m_i}\left( \sum_{k=1}^{K}\mathbb{E}_{Z_i\sim J}[\Delta y_{i,k,j}-\Delta Z_{i,j}]+\mathbb{E}_{Z_i\sim J}[\Delta Z_{i,j}]\right) +\rho\right]\lambda_i  \right\rbrace\\
	&= \exp\left\lbrace -\left[\sum_{j=1}^{m_i}\left( \sum_{k=1}^{K}\mathbb{E}_{Z_i\sim J}[\Delta y_{i,k,j}-\Delta Z_{i,j}]+\mathbb{E}_{Z_i\sim J}[\Delta Z_{i,j}]\right) +\rho \right]\lambda_i \right\rbrace \\
	&\quad\times\lambda_i^{\sum_{j=1}^{m_i}\left( \sum_{k=1}^{K}\Delta\eta_{i,k,j} +\Delta\eta_{i,0,j}\right) +\omega-1}.
\end{align*}
\end{document}